\documentclass[12pt]{amsart}
\usepackage{amsfonts}
\usepackage{amsmath}
\usepackage{color}
\usepackage{graphicx}
\usepackage{hyperref}

\newtheorem{thm}{Theorem}[section]

\newtheorem{lem}[thm]{Lemma}
\newtheorem{prop}[thm]{Proposition}
\newtheorem*{prob*}{Problem}
\newtheorem*{thm*}{Theorem}

\theoremstyle{definition}
\newtheorem{defn}[thm]{Definition}
\newtheorem{example}[thm]{Example}
\newtheorem*{defn*}{Definition}

\newtheorem{remark}[thm]{Remark}

\numberwithin{equation}{section}

\newcommand{\C}{\mathbb C}
\newcommand{\R}{\mathbb R}

\newcommand{\Y}{\mathbb Y}
\newcommand{\Z}{\mathbb Z}

\newcommand{\X}{\mathfrak{X}}

\DeclareMathOperator{\Volume}{Volume}

\DeclareMathOperator{\Schur}{Schur}

\DeclareMathOperator{\const}{const}

\DeclareMathOperator{\Prob}{Prob}

\newcommand{\E}{\mathbb E}

\begin{document}
\title[Products of random matrices as limits of random plane partitions]
{\bf{Product matrix processes as limits of random plane partitions}}
\author{Alexei Borodin}
\address{Department of Mathematics, Massachusetts Institute of Technology, USA, and Institute for Information
Transmission Problems of Russian Academy of Sciences, Russia}\email{borodin@math.mit.edu}
\author{Vadim Gorin}
\address{Department of Mathematics, Massachusetts Institute of Technology, USA, and Institute for Information
Transmission Problems of Russian Academy of Sciences, Russia}\email{vadicgor@gmail.com}
\author{Eugene Strahov}
\address{Department of Mathematics, The Hebrew University of
Jerusalem, Givat Ram, Jerusalem
91904, Israel}\email{strahov@math.huji.ac.il}


\commby{}
\begin{abstract}
We consider a random process with discrete time formed by singular values of products of
truncations of Haar distributed unitary matrices. We show  that this process can be understood as
a scaling limit of the Schur process, which gives  determinantal formulas for (dynamical)
correlation functions and a contour integral representation for the correlation kernel. The
relation with the Schur processes implies that the continuous limit of marginals for
q-distributed plane partitions coincides with the joint law of singular values for products of
truncations of Haar-distributed random unitary matrices. We provide structural reasons for this
coincidence that may also extend to other classes of random matrices.
\end{abstract}
\maketitle

\section{Introduction}
It was observed by many researchers that probability distributions from random matrix theory appear
as limit laws in a variety of problems in statistical mechanics and combinatorics. Probably, the
most known examples of this phenomenon are Ulam's problem for increasing subsequences of random
permutations and domino tilings of the Aztec diamond. We refer the reader to the book by Baik,
Deift, and Suidan \cite{BaikDeiftSuidan} and references therein for a detailed analysis of these
two examples.

In both these problems, distributions of certain key quantities converge to \textit{limiting}
distributions of random matrix theory (in an appropriate scaling limit as the size of random matrices
tends to infinity), in particular to the Tracy-Widom distribution \cite{TW}. However, there are
situations where not only limiting random matrix distributions play a role. It can happen that
joint laws of eigenvalues (or singular values) corresponding to \emph{finite size} random matrix
ensembles arise as scaling limits in a combinatorial or statistical mechanics problem that has no
\emph{a priori} relation
 with random matrices.

 One example is the GUE-corners process investigated by Johansson, Nordenstam \cite{JohanssonNordenstam}
 and Okounkov, Reshetikhin \cite{OR-birth}.
It is easy to define the corners process: start with an infinite random matrix picked from the
Gaussian Unitary Ensemble (GUE) and consider the eigenvalues of its principal corner submatrices.
 As was discovered in
\cite{JohanssonNordenstam}, \cite{OR-birth},  the GUE-corners process can be obtained as a scaling
limit in tiling models. Okounkov and Reshetikhin further suggested a (heuristic) argument towards
the universal appearance of this process. Following this prediction, GUE--corners process was found
in more general tilings models by Gorin, Panova \cite{GP} and the six--vertex model by Gorin
\cite{Gorin_ASM} and Dimitrov \cite{Dimitrov_six}. It can be also linked to the last passage
percolation, see Baryshnikov \cite{Baryshnkiov}, Gravner, Tracy, Widom \cite{GTW}, O'Connell, Yor
\cite{OCY}, Bougerol, Jeuli \cite{BJ}, Adler, van Moerbeke, Wang \cite{AdlerMoerbekeWang}, and
references therein.

\bigskip

The GUE--corners process is  a \emph{determinantal} (see, e.g., Borodin \cite{B-det}) point process
with discrete time obtained using random matrices. If, instead of cutting out corners of a single
matrix, one starts \emph{adding} independent GUE matrices, then the eigenvalues of the sums also
form a determinantal process,
 and the number of matrices in the sum plays the role of  discrete time, see Eynard, Mehta \cite{EynardMehta}.
   Another class of (dynamical) determinantal processes with  discrete time  can be constructed
   from \textit{products} of random matrices, see Strahov \cite{StrahovD},
Akemann and Strahov \cite{AkemannStrahovModels}. (Determinantal processes in products of random
matrices were first discovered in Akemann and Burda \cite{Akemann1}). Such processes are called \emph{product matrix processes}, and they are formed by the squared
singular values of random matrix products. We can use independent complex Gaussian matrices to
obtain a simple example of such processes. Namely, let $G_1$, $\ldots$, $G_m$ be independent
matrices with standard i.i.d.\ complex Gaussian entries. Assume that $G_l$ is of size
$\left(n+\nu_l\right)\times\left(n+\nu_{l-1}\right)$, $\nu_0=0$, $\nu_1\geq 0$, $\ldots$,
$\nu_{m-1}\geq 0$, and for each $l=1,\ldots,m,$ denote by $y_j^l$, $j=1,\ldots,n$, the squared
singular values of the partial product $Y_l=G_l\cdots G_1$. The configuration
$$
\left\{\left(l,y_j^l\right)\vert l=1,\ldots,m; j=1,\ldots,n\right\}
$$
of all these squared singular values generates a random point process on $\left\{1,\ldots,m\right\}\times
\R_{>0}$. It was shown in Strahov \cite{StrahovD} that this process is determinantal (and it can be
viewed as a determinantal process with  discrete time). Paper \cite{StrahovD} gives a contour
integral representation for the correlation kernel, together with its hard edge scaling limit, and
generalizes results obtained in Akemann, Kieburg, and Wei \cite{AkemannKieburgWei}, Akemann, Ipsen,
and Kieburg \cite{AkemannIpsenKieburg}, Kuijlaars and Zhang \cite{KuijlaarsZhang} to the multi-level
situation. A more general class of product matrix processes related to certain multi-matrix models
was introduced and studied in Akemann and Strahov \cite{AkemannStrahovModels}. In this class  the
matrices in the products are no longer independent, but in spite of that the product matrix
processes are still determinantal.

From a different viewpoint, various matrix corners processes studied by Johansson, Nordenstam
\cite{JohanssonNordenstam}, Okounkov, Reshetikhin \cite{OR-birth}, Adler, van Moerbeke, Wang
\cite{AdlerMoerbekeWang}, Forrester, Rains \cite{ForresterRains}, Borodin, Gorin \cite{BG_GFF} were
shown to be continuous limits of special \textit{Schur processes} of Okounkov and Reshetikhin
\cite{Ok-wedge}, \cite{OkounkovReshetikhin}. The discrete time determinantal process formed by the
eigenvalues of sums of independent GUE matrices can be understood as a limit of a special Schur
process as well. It is natural to ask whether product matrix processes also have this property.
Motivated by this question, we construct in this paper a product matrix process using corners of
independent Haar distributed unitary matrices (or truncated unitary matrices). We demonstrate that
this process is a scaling limit of a certain Schur process, which implies determinantal formulas
for (dynamical) correlation functions. Moreover, starting  from the general Okounkov-Reshetikhin
formula  \cite {OkounkovReshetikhin} for the correlation kernels of Schur processes, we derive a
double contour integral representation for the correlation kernel of the product matrix process
formed by truncated unitary matrices. The formula for the correlation kernel we derive in this
paper can be understood as a time-dependent generalization of the result obtained in Kieburg,
Kuijlaars, and Stivigny \cite{KieburgKuijlaarsStivigny} for the squared singular values of matrix
products with truncated unitary matrices.

The fact that the product matrix process formed by truncated unitary matrices is a continuous limit
of the Schur process enables us to prove  Theorem \ref{MainTheorem} below, that says that the
continuous limit of marginals for $q$-distributed (skew) plane partitions coincides with the joint law of
squared singular values for products of corners of Haar-distributed unitary matrices, see Figure \ref{Fig_Plane_partition} for one particular case of the theorem. We consider
Theorem \ref{MainTheorem} as the main result of the present paper.  It demonstrates that, similarly
to the corners process, the time-dependent determinantal processes constructed from products of
truncated unitary matrices appear as  scaling limits in a model of statistical mechanics of a
combinatorial nature. To the best our knowledge, the present paper is the first work relating
products of random matrices with scaling limits of models that have no \emph{a priori} relation to
random matrices.

\begin{figure}[t]
\begin{center}
 {\scalebox{1.0}{\includegraphics{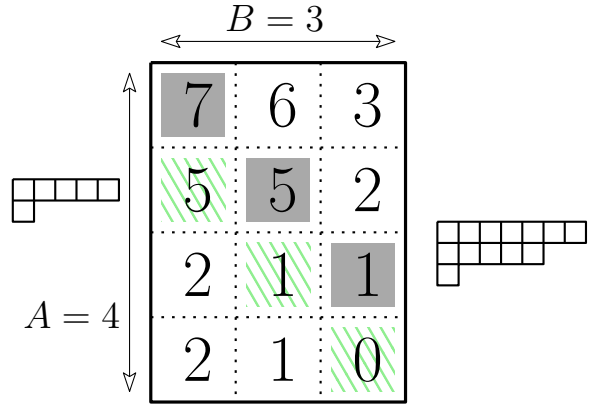}}}
\end{center}
 \caption{Plane partition with $4\times 3$ support corresponds to a sequence of (interlacing) Young diagrams. Two of them,  $5\ge 1\ge 0$ and $7\ge 5\ge 1$, are highlighted. The asymptotic behavior of these two diagrams is related to squared singular values of $T_1$ and $T_2 T_1$, respectively, where $T_1$ is $4\times 3$ truncation of $7\times 7$ random unitary matrix and $T_2$ is $3\times 4$ truncation of $4\times 4$ random unitary matrix. In the notations of Section \ref{Section_skew_PP}, $A=4$, $B=3$,  $\pi=\emptyset$, $p=2$, $\beta_2=1$, $\alpha_1=4$, $\alpha_2=5$.
 \label{Fig_Plane_partition}}
\end{figure}

\medskip

It is also natural to ask about conceptual reasons for such a coincidence. Why should random matrices be
directly related to  statistical mechanics models? For the GUE--corners process, Okounkov and
Reshetikhin \cite{OR-birth} suggested the following heuristic argument: If we start with a discrete
model of statistical mechanics satisfying a certain Gibbs (i.e., conditional uniformity) property,
then one expects the same property to survive in the continuous limit. Olshanski and Vershik
\cite{Olsh_Versh} classified all such Gibbs measures on triangular arrays of reals, and out of them
only the GUE--corners process agrees with the growth conditions implied by the Law of Large Numbers
(limit shape behavior) of the discrete model.

For the products of random matrices we do not dispose of an analogue of the Okounkov--Reshetikhin
argument. However, let us explain the path that led us to the understanding that a similar
connection with combinatorial statistical mechanics is possible. The following fact
\cite{ForresterRains}, \cite{BG_GFF} is easy to prove by comparing the explicit formulas for the
distributions: If $X$ is a corner of Haar-random unitary matrix, then the eigenvalues of $X X^*$
(this is often called Jacobi or MANOVA ensemble) are distributed as a continuous limit of a
\emph{Schur measure} with two principal specializations. One could argue that this is an instance
of the semiclassical asymptotics common in representation theory.  Next, we need to understand
what happens with these eigenvalues when $X$ is multiplied by another similar matrix. It is known
that multiplication of (real/complex/quaternion at $\beta=1,2,4$) matrices is intimately linked to
multiplication of corresponding Jack (=zonal) polynomials, which become Schur polynomials in the
case of the complex field ($\beta=2$) that we discuss here. This is discussed by Macdonald
\cite[Chapter VII]{Macdonald}, Forrester \cite[Section 13.4.3]{For}, and more recently used, e.g.,
by Kieburg, Kosters \cite{KK} and by Gorin, Marcus \cite{GorinMarcus}. If we consider a
 version of the multidimensional Fourier transform for the Schur measures (the appropriate version
was introduced by Gorin, Panova \cite{GP}, Bufetov, Gorin \cite{BG_Schur} under the name
\emph{Schur generating functions}), then being a Schur measure or its continuous limit is
equivalent to the factorization of this transform into a product of  one variable functions. Since
such factorization is preserved under multiplication, the squared singular values of products of
random matrices have to be described by the Schur processes.

We detail how this argument works in the simplest case of $2\times 2$ matrices in Section
\ref{Section_spherical}. For the proof of our main statements, Proposition
\ref{PropositionGeneralLimit} and Theorem \ref{MainTheorem}, we choose in Sections
\ref{Section2}-\ref{Section7} another path, which is more direct (and leads to a more general
result) but, perhaps, more mysterious. Let us remark that while the arguments of Section
\ref{Section_spherical} admit an immediate generalization to products of real and quaternion
matrices, yielding their representation as limits of \emph{Macdonald processes}, for the proofs of
Sections \ref{Section2}-\ref{Section7} such a generalization is unclear.

\bigskip

The paper is organized as follows. In Section \ref{Section2} we introduce notation and present the
main results. In particular, Proposition \ref{THEOREMCorrelationKernel} gives a formula for the
correlation kernel of the product matrix process associated with truncated unitary matrices,
Proposition \ref{PropositionGeneralLimit} shows that this product matrix process can be understood
as a continuous limit of a special Schur process, and Theorem \ref{MainTheorem} presents our result
on convergence of marginals of $q$-distributed plane partitions to this product matrix process.
Sections \ref{Section2}-\ref{Section7} contain the proofs of our statements. In Section
\ref{Section_spherical} we sketch another way to prove our main results by exploiting symmetric
functions and zonal polynomials. Finally, Appendix gives a second proof of  Proposition
\ref{THEOREMCorrelationKernel} based on the Eynard-Mehta theorem.

\medskip

 \textbf{Acknowledgements}. We are
very grateful to Mario Kieburg and Leonid Petrov for discussions. All three authors were partially
supported by the NSF grant DMS-1664619. A.B.~was partially supported by the NSF grant  DMS-1607901.
V.G.~was partially supported by the by NEC Corporation Fund for Research in Computers and
Communications and by the Sloan Research Fellowship.

\section{Notation and statement of results}\label{Section2}

\subsection{Product matrix processes with truncated unitary matrices}
\label{SectionTruncatedProcess}
Let $G_1$, $\ldots$, $G_p$ be  matrices with random complex entries, and assume that each matrix $G_k$, $k\in\left\{1,\ldots,p\right\}$,
is of size $N_{k}\times N_{k-1}$. Set
$$
X(k)=G_k\ldots G_1,\;\;\; k\in\left\{1,\ldots,p\right\}.
$$
If $n=N_0$, then for each $k$, $k\in\left\{1,\ldots,p\right\}$, $X^*(k)X(k)$ are random matrices of the same size $n\times n$.
Denote by $x_j^k$ the $j$th largest eigenvalue of $X^*(k)X(k)$.
The configuration of all these eigenvalues,
\begin{equation}\label{Configuration}
\left\{\left(k,x_j^k\right)\biggl\vert k=1,\ldots,p; j=1,\ldots,n\right\},
\end{equation}
forms a point process on $\left\{1,\ldots,p\right\}\times\R_{>0}$. This point process is called
the \textit{product matrix process} associated with the random matrices $X(1)$, $X(2)$, $\ldots$, $X(p)$.

Here we consider a product matrix process constructed from a collection of truncated unitary matrices. Namely, let $U_1$, $\ldots$,
$U_l$, $U_{l+1}$,$\ldots$, $U_{l+p-1}$ be independent Haar distributed unitary matrices. We assume that the size of each matrix $U_j$, $1\leq j\leq p+l-1$, is equal to $m_j$.
Recall that if $U$ is a $m\times m$ matrix, and the integers $k$, $r$ are chosen such that $1\leq k, r\leq m$,  the submatrix $T$ of $U$ defined by
$$
T=\left(\begin{array}{ccc}
          U_{1,1} & \ldots & U_{1,r} \\
          \vdots &  &  \\
          U_{k,1} & \ldots & U_{k,r}
        \end{array}
\right)
$$
is called a $k\times r$ truncation of $U$. Now, for $1\leq j\leq p+l-1$ let $T_j$ be the truncation of $U_j$ of size $\left(n+\nu_j\right)\times \left(n+\nu_{j-1}\right)$.
We agree that $\nu_0=0$, and assume that the positive integers $n$, $\nu_1$, $\ldots$, $\nu_{l+p-1}$ are chosen in such a way that the conditions
\begin{equation}\label{Condition1}
m_1\geq 2n+\nu_1,
\end{equation}
and
\begin{equation}\label{Condition2}
m_j\geq n+\nu_j+1,\;\;\; 2\leq j\leq p+l-1,
\end{equation}
are satisfied. Denote by $x^1=\left(x_1^1,\ldots,x_n^1\right)$ the vector of the squared singular values of the product matrix $T_l\ldots T_1$,
and for $2\leq j\leq p$ denote by $x^j=\left(x^j_1,\ldots,x^j_n\right)$ the vector of the squared singular values of the product matrix
$T_{j+l-1}\ldots T_1$. Configurations $\left\{\left(k,x_j^k\right)\biggl\vert k=1,\ldots,p; j=1,\ldots,n\right\}$ form a point process on
$\left\{1,\ldots,p\right\}\times\R_{>0}$. We will refer to this point process as to the \textit{product matrix process associated with truncated unitary matrices}.
We say that the product matrix $T_l\ldots T_1$ determines the \textit{initial conditions} of the product matrix process associated with truncated unitary matrices.
The numbers $n$, $m_j$, $\nu_j$ will be called the \textit{parameters} of the product matrix process associated with truncated unitary matrices.
\begin{prop}\label{PropositionProductTruncatedProcess}
Consider the product matrix process associated with truncated unitary matrices, and let
$x_1^k\leq\ldots\leq x_n^k$; $k=1,\ldots,p,$ denote the set of the squared singular values of the product matrix
$T_{k+l-1}\ldots T_1$.
The joint probability distribution of $\left(x_1^k,\ldots,x_n^k\right)$ is given by
\begin{equation}\label{TruncatedProcessJD}
\begin{split}
&\frac{1}{Z_{n,p+l}}
\triangle\left(x^p\right)\\
&\times\prod\limits_{r=1}^{p-1}\det\left[\left(x_j^{r+1}\right)^{\nu_{l+r}}\left(x_k^{r}-x_j^{r+1}\right)_+^{m_{l+r}-n-\nu_{l+r}-1}\left(x_k^{r}\right)^{n-m_{l+r}}\right]_{k,j=1}^n\\
&\times\det\left[w_k^{(l)}\left(x_j^1\right)\right]_{k,j=1}^n
dx^1\ldots dx^n,
\end{split}
\end{equation}
where $(x-y)_+=\max\left(0,x-y\right)$, the Vandermonde determinant $\triangle\left(x^p\right)$ is defined by
$\triangle\left(x^p\right)=\prod\limits_{1\leq i<j\leq n}\left(x_j^p-x_i^p\right)$,
for $1\leq l\leq p$ we write $dx^l=dx^l_1\ldots dx^l_n$, $Z_{n,p+l}$ is a normalization constant, and
$w_k^{(l)}(x)$ is a sequence of weight functions.
The normalization constant $Z_{n,p+l}$ can be written explicitly as
\begin{equation}
Z_{n,p+l}=\frac{\prod\limits_{j=1}^n\Gamma\left(m_1-2n-\nu_1+j\right)\Gamma(j)
\prod\limits_{k=2}^{p+l-1}\left(\Gamma\left(m_k-n-\nu_k\right)\right)^n}{\prod\limits_{k=1}^{p+l-1}\prod\limits_{j_k=1}^{m_k-n-\nu_k}
\left(j_k+\nu_k\right)_n}.
\end{equation}
Here $(a)_m=a(a+1)\ldots (a+m-1)$ stands for the Pochhammer symbol.
The function $w_k^{(l)}(x)$ can be expressed as a Meijer $G$-function,
\begin{equation}\label{TheWeightFunctions}
\begin{split}
&w_k^{(l)}(x)=c_lG_{l,l}^{l,0}\left(\begin{array}{cccc}
                                     m_l-n, & \ldots, & m_2-n, & m_1-2n+k \\
                                     \nu_l, & \ldots, & \nu_2, & \nu_1+k-1
                                   \end{array}
\biggl|x\right)\\
&=\frac{c_l}{2\pi i}\int\limits_C\frac{\Gamma\left(\nu_1+k-1+s\right)\prod_{j=2}^l\Gamma\left(\nu_j+s\right)}{\Gamma\left(m_1-2n+k+s\right)\prod_{j=2}^l
\Gamma\left(m_j-n+s\right)}x^{-s}ds,\;\;\; 0<x<1.
\end{split}
\end{equation}
In this formula $C$ denotes a positively oriented contour in the complex $s$-plane that starts and ends at $-\infty$ and encircles
the negative real axis. The constant $c_l$ in the formula for $w_k^{(l)}(x)$ can be written as
\begin{equation}\label{cl}
c_l=\Gamma\left(m_1-2n-\nu_1+1\right)\prod\limits_{j=2}^l\Gamma\left(m_j-n-\nu_j\right).
\end{equation}
\end{prop}
\begin{remark} The  Meijer $G$-function in equation (\ref{TheWeightFunctions}) is equal to
zero for $x\geq 1$. Correspondingly, we set the weight functions  $w_k^{(l)}(x)$  to zero for $x\geq 1$.
\end{remark}
\begin{remark}
The righthand side of equation (\ref{TheWeightFunctions}) can be written as
$$\frac{c_l'}{2\pi i}
\int\limits_C\frac{\left(s\right)_{\nu_1+k-1}\prod_{j=2}^l\left(s\right)_{\nu_j}}{\left(s\right)_{m_1-2n+k}\prod_{j=2}^l
\left(s\right)_{m_j-n}}x^{-s}ds,
$$
where
$$
c_l'=B\left(m_1-2n-\nu_1+1,\nu_1+k-1\right)\prod\limits_{j=2}^lB\left(m_j-n-\nu_j,\nu_j\right),
$$
and $B(x,y)$ stands for the Beta function.
\end{remark}
Our next result provides explicit formulae for the correlation functions of the product matrix process
associated with truncated unitary matrices.
\begin{prop}\label{THEOREMCorrelationKernel}
The product matrix process with truncated unitary matrices is a determinantal process on
$\left\{1,\ldots,p\right\}\times \R_{>0}$. Its correlation kernel, $K_{n,p,l}(r,x;s,y)$
(where $r,s\in\left\{1,\ldots,p\right\}$, and $x, y\in\R_{>0}$) can be written as
\begin{equation}\label{CorrelationKernelProductProcess}
\begin{split}
&K_{n,p,l}(r,x;s,y)=-\frac{1}{x}
G_{s-r,s-r}^{s-r,0}\left(\begin{array}{ccc}
                           m_{r+l}-n, & \ldots, & m_{s+l-1}-n \\
                           \nu_{r+l}, & \ldots, & \nu_{s+l-1}
                         \end{array}
\biggl\vert\frac{y}{x}\right)\mathbf{1}_{s>r}\\
&+
\frac{1}{\left(2\pi i\right)^2}\oint\limits_{C_t}dt\oint\limits_{C_{\zeta}}d\zeta
\frac{\prod\limits_{a=0}^{s+l-1}\Gamma\left(\nu_a+\zeta+1\right)}{\prod\limits_{a=0}^{r+l-1}\Gamma\left(\nu_a+t+1\right)}
\frac{\prod\limits_{a=0}^{r+l-1}\Gamma\left(m_a-n+t+1\right)}{\prod\limits_{a=0}^{s+l-1}\Gamma\left(m_a-n+\zeta+1\right)}
\frac{x^ty^{-\zeta-1}}{\zeta-t},
\end{split}
\end{equation}
where $C_t$ is a closed contour in the complex $t$-plane encircling the interval $[0,n-1]$ once in the positive direction,
$C_{\zeta}$ is a positively oriented closed contour in the complex $\zeta$-plane encircling once an interval containing all the points
$-\left(1+\nu_1\right)$, $\ldots$, $-\left(m_1-n\right)$; $\ldots$; $-\left(1+\nu_{s+l-1}\right)$, $\ldots$, $-\left(m_{s+l-1}-n\right)$, which does not intersect $C_t$.
In the formula above
it is understood that $m_0=\nu_0=0$.
\end{prop}
We note that at $r=s=p$ and $l=1$ the formula for the correlation kernel stated in Proposition \ref{THEOREMCorrelationKernel}
turns into that derived in Kieburg, Kuijlaars, and Stivigny \cite{KieburgKuijlaarsStivigny}, Proposition 2.7.

In what follows we will give two proofs of Proposition \ref{THEOREMCorrelationKernel}. The first proof will use the fact
that the process defined by equation (\ref{TruncatedProcessJD}) can be understood as a continuous limit of a special Schur process, see Proposition \ref{PropositionGeneralLimit}
below.
Since the Schur processes are determinantal, this will imply determinantal formulae for the correlation functions. As for the explicit formula for
the correlation kernel (see equation (\ref{CorrelationKernelProductProcess})), it will be obtained from the general Okounkov-Reshetikhin formula \cite{OkounkovReshetikhin} for correlation kernels
of the Schur processes by a certain limiting procedure in Section \ref{Section 6}. The second proof will be based on the observation that the density of the product matrix process with truncated
unitary matrices can be written as a product of determinants. This will enable us to apply the result by Eynard
and Mehta \cite{EynardMehta}, and to give a formula for the correlation functions, see the Appendix. This second argument is similar to the proof of Borodin and Rains \cite{BorodinRains} of the determinantal structure of Schur processes.

\subsection{Convergence of the Schur process}
\label{TheSchurProcesses}
In this Section we use the notation of Macdonald \cite{Macdonald}, and follow  Refs. \cite{OkounkovReshetikhin, Borodin, BorodinGorin, BorodinRains}.

Let $\Lambda$ be the algebra of symmetric functions in countably many variables $z_1,z_2,\dots$ . We use two sets of generators of $\Lambda$: power sums $p_k$ and complete homogeneous symmetric functions $h_k$, $k=1,2,\dots$, defined through
$$
 p_k=\sum_{i=1}^{\infty} (z_i)^k,\quad h_k=\sum_{i_1\le i_2\le \dots\le i_k} z_{i_1} z_{i_2}\cdots z_{i_k}.
$$
We recall that the Schur functions $s_\lambda$ form a basis of $\Lambda$ when $\lambda$ varies over all Young diagrams (or partitions). We also use skew Schur functions $s_{\lambda/\mu}$ labeled by two Young diagrams $\lambda$ and $\mu$.

 A specialization $\varrho$ of $\Lambda$ is an algebra homomorphism of $\Lambda$ to $\C$.
A specialization $\varrho$ of $\Lambda$ is called nonnegative if it takes non-negative values on the Schur functions, see, e.g., Borodin \cite{Borodin}, Section 1
for a detailed discussion of nonnegative specializations of the algebra of symmetric functions. The application of a specialization $\varrho$ to $f\in\Lambda$ will be denoted as $f(\varrho)$. The trivial specialization $\emptyset$ of $\Lambda$ takes value $1$ at the constant function $1\in\Lambda$, and
takes value $0$ at any homogeneous $f\in\Lambda$ of degree $\geq 1$. In particular $s_\lambda(\emptyset)=0$ unless $\lambda=\emptyset$, and $s_{\lambda/\mu}(\emptyset)=0$ unless $\lambda=\mu$.

 In this paper we will only use the simplest Schur positive specializations parameterized by arbitrary $m=1,2,\dots$ and $m$--tuple of positive reals $(\alpha_1,\dots,\alpha_m)\in\mathbb R_{>0}^m$. We denote it $\varrho=(\alpha_1,\dots,\alpha_m)$ and set
$$
 p_k(\varrho)=p_k(\alpha_1,\dots,\alpha_m)=(\alpha_1)^k+(\alpha_2)^k+\dots (\alpha_m)^k.
$$
Equivalently, this specialization can be encoded by its generating function:
$$
 H(\varrho; u):=1+\sum_{k=1}^{\infty} h_k(\varrho) u^k=\exp\left(\sum_{k=1}^{\infty} \frac{p_k(\varrho)u^k}{k}\right) =\prod_{i=1}^{m} \frac{1}{1-\alpha_i u},
$$
where the second identity is the algebraic relation between generators $p_k$ and $h_k$.

The specializations that we use are often given by geometric series $\varrho=(q^{t},q^{t+1},\dots,q^{s})$. When $t>s$, the geometric series is empty and $\varrho$ becomes the trivial specialization.

\begin{defn}
Let $p$ be a natural number, and let $\varrho_0^+,\ldots,\varrho_{p-1}^+,\varrho_1^-,\ldots,\varrho_p^-$ be nonnegative specializations of $\Lambda$.  The Schur process of rank $p$ is a probability measure on sequences of Young diagrams
$$
\lambda^{(1)},\mu^{(1)},\lambda^{(2)},\mu^{(2)},\ldots,\lambda^{(p-2)},\mu^{(p-2)},\lambda^{(p-1)}, \mu^{(p-1)},\lambda^{(p)}
$$
parameterized by $2p$ Schur-positive specializations of the algebra of symmetric functions
given by
\begin{equation}\label{SchurProcess}
\begin{split}
&\Prob\left(\lambda^{(1)},\mu^{(1)},\lambda^{(2)},\mu^{(2)},\ldots,\lambda^{(p-2)},\mu^{(p-2)},\lambda^{(p-1)},\mu^{(p-1)},\lambda^{(p)}\right)\\
&=\frac{1}{Z_{\Schur}}s_{\lambda^{(1)}}\left(\varrho_0^+\right)s_{\lambda^{(1)}/\mu^{(1)}}\left(\varrho_1^-\right)
s_{\lambda^{(2)}/\mu^{(1)}}\left(\varrho_1^+\right)s_{\lambda^{(2)}/\mu^{(2)}}\left(\varrho_2^-\right)
s_{\lambda^{(3)}/\mu^{(2)}}\left(\varrho_2^+\right)\\
&\times\ldots\times s_{\lambda^{(p-1)}/\mu^{(p-1)}}\left(\varrho_{p-1}^-\right)s_{\lambda^{(p)}/\mu^{(p-1)}}\left(\varrho_{p-1}^+\right)
s_{\lambda^{(p)}}\left(\varrho_p^-\right).
\end{split}
\end{equation}
Here $Z_{\Schur}$ is a normalization constant.
\end{defn}
Since $s_{\lambda/\mu}\equiv 0$ unless $\mu\subset\lambda$, the Schur process lives on the following configurations of Young diagrams
$$
\emptyset\subset\lambda^{(1)}\supset\mu^{(1)}\subset\lambda^{(2)}\supset\mu^{(2)}\subset\lambda^{(3)}\supset\ldots
\subset\lambda^{(p-1)}\supset\mu^{(p-1)}\subset\lambda^{(p)}\supset\emptyset.
$$
For skew Schur functions we have the following summation formulae
\begin{equation}\label{SummationFormula1}
\sum\limits_{\mu\in\Y}s_{\mu/\lambda}(\varrho)s_{\mu/\nu}(\varrho')=H(\varrho;\varrho')\sum\limits_{\kappa\in\Y}s_{\lambda/\kappa}(\varrho')s_{\nu/\kappa}(\varrho),
\end{equation}
and
\begin{equation}\label{SummationFormula2}
\sum\limits_{\nu\in\Y}s_{\lambda/\nu}(\varrho)s_{\nu/\mu}(\varrho')=s_{\lambda/\mu}(\varrho,\varrho'),
\end{equation}
see Macdonald \cite{Macdonald}, Section I.5, equation (5.9), and Example  I.5.26(1).
Here
$$
H\left(\varrho;\varrho'\right)=\exp\left(\sum\limits_{k=1}^{\infty}\frac{p_k(\varrho)p_k(\varrho')}{k}\right),
$$
and the values of the symmetric functions under the union specialization $(\varrho,\varrho')$ are determined by the power sum values given by
$$
p_k\left(\varrho,\varrho'\right)=p_k\left(\varrho\right)+p_k\left(\varrho'\right).
$$
Hence,  for specializations $\varrho_1,\ldots,\varrho_k,\varrho_1',\ldots,\varrho_m'$
we have
$$
H\left(\varrho_1,\ldots,\varrho_k;\varrho'_1,\ldots,\varrho_m'\right)=\prod\limits_{i=1}^k\prod\limits_{j=1}^m
H\left(\varrho_i;\varrho_j'\right).
$$
\begin{prop}\label{PropositionGeneralLimit} Consider the Schur process defined by the probability measure (\ref{SchurProcess}).
Assume that the  specializations
$\varrho_0^+$, $\ldots$, $\varrho_{p-1}^+$ of the Schur process are defined by
  \begin{equation}\label{varrho0specialization}
  \varrho_0^+=\left(e^{-\left(1+\nu_1\right)\epsilon},e^{-\left(2+\nu_1\right)\epsilon},\ldots,e^{-\left(m_1-n\right)\epsilon};
  \ldots;e^{-\left(1+\nu_l\right)\epsilon},e^{-\left(2+\nu_l\right)\epsilon},\ldots,e^{-\left(m_l-n\right)\epsilon}\right),
  \end{equation}
  \begin{equation}\label{varrho1specialization}
  \varrho_1^+=\left(e^{-\left(1+\nu_{l+1}\right)\epsilon},e^{-\left(2+\nu_{l+1}\right)\epsilon},\ldots,e^{-\left(m_{l+1}-n\right)\epsilon}\right),
  \end{equation}
  $$
  \vdots
  $$
  \begin{equation}\label{varrhopminus1specialization}
  \varrho_{p-1}^+=\left(e^{-\left(1+\nu_{l+p-1}\right)\epsilon},e^{-\left(2+\nu_{l+p-1}\right)\epsilon},\ldots,e^{-\left(m_{l+p-1}-n\right)\epsilon}\right).
  \end{equation}
  The  specialization $\varrho_{p}^{-}$ is defined by
  \begin{equation}\label{negativespecialization}
  \varrho_{p}^{-}=\left(1,e^{-\epsilon},\ldots, e^{-(n-1)\epsilon}\right),
  \end{equation}
and all the other  specializations $\varrho_1^{-}$, $\ldots$, $\varrho_{p-1}^{-}$ are trivial.
With these specializations the Schur process lives on the point configurations
\begin{equation}\label{CSP}
\left\{\left(1,\lambda_i^{(1)}-i\right)\right\}_{i=1}^{n}\cup\ldots\cup\left\{\left(p,\lambda_i^{(p)}-i\right)\right\}_{i=1}^{n},
\end{equation}
and each  $\lambda^{(k)}$, $1\leq k\leq p$ has at most $n$ nonzero parts almost surely.
Set
\begin{equation}\label{rconfigurations11}
x_j^k=e^{-\epsilon\lambda_j^{(k)}},\; k=1,\ldots,p;\; j=1,\ldots,n.
\end{equation}
Then the Schur process induces a point process on $\left\{1,\ldots,p\right\}\times\R_{>0}$, and this process
is formed by  the configurations
\begin{equation}
\left\{\left(k,x_j^k\right)\biggl|k=1,\ldots,p;\; j=1,\ldots,n\right\}.
\end{equation}
As $\epsilon\rightarrow 0$, the point process formed by configurations (\ref{rconfigurations11}) converges to the product matrix process
associated with truncated unitary matrices, as defined in Section \ref{SectionTruncatedProcess}.
\end{prop}
\begin{remark}
\label{Remark_size_restriction}
 We prove Proposition \ref{PropositionGeneralLimit} only under the assumption $m_1\ge 2n+\nu_1$ of \eqref{Condition1}. Although it is very plausible that the statement is true without this condition, we do not address the more general case in this text. A technical difficulty is that without \eqref{Condition1} we cannot use the result of Proposition \ref{PropositionProductTruncatedProcess} directly; in particular, the constant $c_l$ of \eqref{cl} is infinite.
\end{remark}
\subsection{Random skew plane partitions, and products of truncated unitary matrices}
\label{Section_skew_PP}
Let $A$ and $B$ be two natural numbers, and denote by $B^A$ the $A\times B$ rectangle.
Let $\pi$ be a Young diagram such that  $\pi\subset B^A$. A skew plane partition $\Pi$ with support
$B^A/\pi$ is a filling of all boxes of $B^A/\pi$ such that $\Pi_{i,j}\geq \Pi_{i+1,j}$
and $\Pi_{i,j}\geq\Pi_{i,j+1}$ for all meaningfull values of $i$ and $j$. Here we assume that $\Pi_{i,j}$ is located in the $i\text{th}$
row and $j\text{th}$ column of $B^A$. The volume of a skew plane partition $\Pi$ is defined by
\begin{equation}\label{Volume}
\Volume\left(\Pi\right)=\sum\limits_{i,j}\Pi_{i,j}.
\end{equation}
Given a parameter $0<q<1$, define a probability measure on the set of all plane partitions $\Pi$ with support
$B^A/\pi$ by setting
\begin{equation}\label{PartitionWeight}
\Prob\left\{\Pi\right\}\sim q^{\Volume\left(\Pi\right)}.
\end{equation}
For a skew plane partition $\Pi$ we define Young diagrams $\lambda^{(k)}\left(\Pi\right)$, $1\leq k\leq A+B+1$, through
$$
\lambda^{(k)}\left(\Pi\right)=\left\{\Pi_{i,i+k-A-1}\biggl|\left(i,i+k-A-1\right)\in B^A/\pi\right\}.
$$
 Note that $\lambda^{(1)}=\lambda^{(A+B+1)}=\emptyset$.
Also, define
$$
\mathcal{L}\left(\pi\right)=\left\{A+\pi_i-i+1\biggl|i=1,\ldots,A\right\}.
$$
This is a subset of $\left\{1,\ldots,A+B+1\right\}$ containing $A$ points, and all such subsets are in bijection with the Young diagrams (or partitions) $\pi$
contained in the box $B^A$.

It is not hard to see that the set of all skew plane partitions with support $B^A/\pi$ consists of sequences
$\left(\lambda^{(1)},\ldots,\lambda^{(A+B+1)}\right)$ with
$$
\lambda^{(1)}=\lambda^{(A+B+1)}=\emptyset,
$$
\begin{equation}
\label{eq_interlacing_PP}
\lambda^{(j)}\prec\lambda^{(j+1)}\;\text{if}\;j\in\mathcal{L}(\lambda),\;\;\;\lambda^{(j)}\prec\lambda^{(j+1)}\;\text{if}\;j\notin\mathcal{L}(\lambda),
\end{equation}
where notation $\mu\prec\nu$  means that $\mu$ and $\nu$ interlace, that is
$$
\nu_1\geq \mu_1\geq\nu_2\geq\mu_2\geq\nu_3\geq\ldots .
$$
We refer to Figure \ref{Fig_Plane_partition} for an illustration of $A=4$, $B=3$, $\pi=\emptyset$ case.
Moreover, we have
$$
\sum\limits_{j=1}^{A+B+1}\left|\lambda^{(j)}\right|=\Volume\left(\Pi\right),
$$
where $|\mu|$ denotes the number of boxes in the Young diagram $\mu$.

The probability measure on the set of all skew plane partitions $\Pi$ with support
$B^A/\pi$ and defined by equation (\ref{PartitionWeight}) induces a probability measure on sequences
$\left(\lambda^{(1)},\ldots,\lambda^{(A+B+1)}\right)$. It is known (see Okounkov and Reshetikhin \cite{OkounkovReshetikhin,OkounkovReshetikhin_skew})
that this probability measure can be understood as the Schur process defined in Section \ref{TheSchurProcesses}
by equation (\ref{SchurProcess}) with the rank $p=A+B+1$, and nonnegative specializations
$\left\{\varrho_i^+\right\}_{i=0}^{p-1}$, $\left\{\varrho_i^-\right\}_{i=1}^p$ defined by
$$
H\left(\varrho_0^+;u\right)=H\left(\varrho_{A+B+1}^-;u\right)=1,
$$
$$
H\left(\varrho_j^+;u\right)=\left\{
                              \begin{array}{ll}
                                \dfrac{1}{1-q^{-j}u}, & j\in\mathcal{L}(\pi), \\
                                1, & j\notin\mathcal(\pi);
                              \end{array}
                            \right.
\;\;\;\;\;
H\left(\varrho_j^-;u\right)=\left\{
                              \begin{array}{ll}
                               1, & j\in\mathcal{L}(\pi),\\
                               \dfrac{1}{1-q^{j}u}, & j\notin\mathcal{L}(\pi).
                                 \end{array}
                            \right.
$$
Note that for any two neighboring specializations $\varrho_k^{-}$, $\varrho_k^{+}$ defined above at least one is trivial, and each $\mu^{(j)}$
coincides either with $\lambda^{(j)}$ or $\lambda^{(j+1)}$. The only nontrivial specializations are one variable specializations $\rho=(\alpha)$ with $\alpha=q^{\pm j}$. A basic property of skew Schur functions is that $s_{\lambda/\mu}(\alpha)=0$ unless $\mu\prec \lambda$; this implies interlacing conditions \eqref{eq_interlacing_PP}.

Let $\pi\in B^A$. The set $\left\{1,\ldots,A+B+1\right\}$ enumerates the intersections of the boundary of the skew diagram $B^A/\pi$
with the square grid, as  shown on Figure \ref{Fig_boundary}. Denote by $l$ the number of vertical segments of
of the boundary of the skew diagram  $B^A/\pi$.
\begin{center}
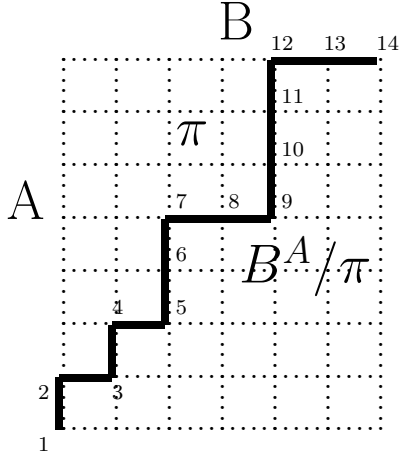
\begin{figure}[h]
\setlength{\unitlength}{4pt}
\begin{picture}(80,40)
\multiput(0,0)(0,1){36}{.}
\multiput(5,0)(0,1){36}{.}
\multiput(10,0)(0,1){36}{.}
\multiput(15,0)(0,1){36}{.}
\multiput(20,0)(0,1){36}{.}
\multiput(25,0)(0,1){36}{.}
\multiput(30,0)(0,1){36}{.}
\multiput(0,0)(1,0){30}{.}
\multiput(0,5)(1,0){30}{.}
\multiput(0,10)(1,0){30}{.}
\multiput(0,15)(1,0){30}{.}
\multiput(0,20)(1,0){30}{.}
\multiput(0,25)(1,0){30}{.}
\multiput(0,30)(1,0){30}{.}
\multiput(0,35)(1,0){30}{.}
\linethickness{1mm}
\put(0,0){\line(0,1){5}}
\put(0,5){\line(1,0){5}}
\put(5,5){\line(0,1){5}}
\put(5,10){\line(1,0){5}}
\put(10,10){\line(0,1){5}}
\put(10,15){\line(0,1){5}}
\put(10,20){\line(1,0){5}}
\put(15,20){\line(1,0){5}}
\put(20,20){\line(0,1){5}}
\put(20,25){\line(0,1){5}}
\put(20,30){\line(0,1){5}}
\put(20,35){\line(1,0){5}}
\put(25,35){\line(1,0){5}}
\begin{tiny}
\put(-2,-2){1}
\put(-2,3){2}
\put(5,3){3}
\put(5,11){4}
\put(11,11){5}
\put(11,16){6}
\put(11,21){7}
\put(16,21){8}
\put(21,21){9}
\put(21,26){10}
\put(21,31){11}
\put(20,36){12}
\put(25,36){13}
\put(30,36){14}
\end{tiny}
\begin{LARGE}
\put(-5,20){A}
\put(15,37){B}
\put(11,27){$\pi$}
\put(17,14){$B^A/\pi$}
\end{LARGE}
\end{picture}
\caption{The set $\left\{1,\ldots,A+B+1\right\}$ enumerates the boundary of the skew diagram $B^A/\pi$. In this example
$A=7$, $B=6$, and $l=4$. \label{Fig_boundary}}
\end{figure}
\end{center}
Let $\left\{\beta_1,\ldots,\beta_{2l-1}\right\}$ be a subset of $\left\{1,\ldots,A+B+1\right\}$, where the  numbers $\beta_1$, $\ldots$, $\beta_{2l-1}$
parameterize the vertical segments of the boundary of the skew diagram  $B^A/\pi$, see Figure \ref{Figure_corners}.
For example, for the Young diagram $\pi$ on Figure \ref{Fig_boundary}  we have $\beta_1=12$, $\beta_2=9$, $\beta_3=7$, $\beta_4=5$,
$\beta_5=4$, $\beta_6=3$, $\beta_7=2$ and $\beta_8=1$.  Now, assume that $a_1\geq p$, and pick $p$ numbers $\alpha_1$, $\ldots$, $\alpha_p$ such that
$\beta_2\le \alpha_1<\ldots<\alpha_p<\beta_1$, see Figure \ref{Figure_corners}.
\begin{center}
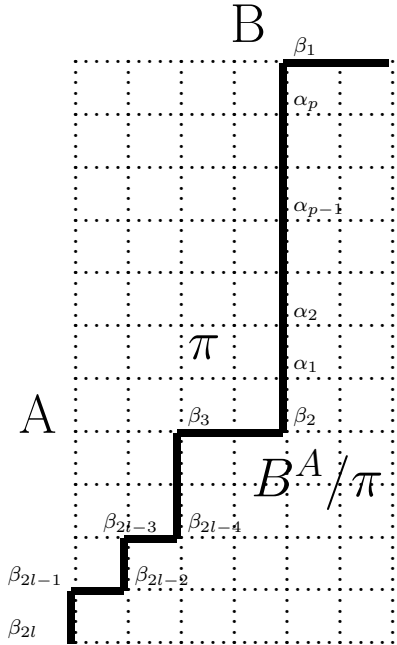
\begin{figure}[h]
\setlength{\unitlength}{4pt}
\begin{picture}(80,60)
\multiput(0,0)(0,1){56}{.}
\multiput(5,0)(0,1){56}{.}
\multiput(10,0)(0,1){56}{.}
\multiput(15,0)(0,1){56}{.}
\multiput(20,0)(0,1){56}{.}
\multiput(25,0)(0,1){56}{.}
\multiput(30,0)(0,1){56}{.}
\multiput(0,0)(1,0){30}{.}
\multiput(0,5)(1,0){30}{.}
\multiput(0,10)(1,0){30}{.}
\multiput(0,15)(1,0){30}{.}
\multiput(0,20)(1,0){30}{.}
\multiput(0,25)(1,0){30}{.}
\multiput(0,30)(1,0){30}{.}
\multiput(0,35)(1,0){30}{.}
\multiput(0,40)(1,0){30}{.}
\multiput(0,45)(1,0){30}{.}
\multiput(0,50)(1,0){30}{.}
\multiput(0,55)(1,0){30}{.}
\linethickness{1mm}
\put(0,0){\line(0,1){5}}
\put(0,5){\line(1,0){5}}
\put(5,5){\line(0,1){5}}
\put(5,10){\line(1,0){5}}
\put(10,10){\line(0,1){5}}
\put(10,15){\line(0,1){5}}
\put(10,20){\line(1,0){5}}
\put(15,20){\line(1,0){5}}
\put(20,20){\line(0,1){5}}
\put(20,25){\line(0,1){30}}
\put(20,55){\line(1,0){10}}
\begin{tiny}
\put(-6,1){$\beta_{2l}$}
\put(-6,6){$\beta_{2l-1}$}
\put(6,6){$\beta_{2l-2}$}
\put(3,11){$\beta_{2l-3}$}
\put(11,11){$\beta_{2l-4}$}
\put(11,21){$\beta_{3}$}
\put(21,21){$\beta_{2}$}
\put(21,26){$\alpha_{1}$}
\put(21,31){$\alpha_{2}$}
\put(21,41){$\alpha_{p-1}$}
\put(21,51){$\alpha_{p}$}
\put(21,56){$\beta_{1}$}
\end{tiny}
\begin{LARGE}
\put(-5,20){A}
\put(15,57){B}
\put(11,27){$\pi$}
\put(17,14){$B^A/\pi$}
\end{LARGE}
\end{picture}
\caption{The parametrization of the the vertical segments of the boundary of $B^A/\pi$ by $\beta_1$, $\ldots$, $\beta_{2l}$, and the choice of $\alpha_1$, $\ldots$, $\alpha_p$. \label{Figure_corners}}
\end{figure}
\end{center}
Consider the sequence $\left(\lambda^{\left(\alpha_1\right)},\ldots,\lambda^{\left(\alpha_p\right)}\right)$ of random Young diagrams associated with
a random skew plane partition $\Pi$ whose support is $B^A/\pi$, and whose weight is proportional to $q^{\Volume(\Pi)}$. By assigning to this sequence the point configuration
\begin{equation}\label{PoinConfiguration}
\left\{\left(1,\lambda_i^{\left(\alpha_1\right)}-i\right)\right\}_{i\geq 1}\cup\ldots\cup\left\{\left(p,\lambda_i^{\left(\alpha_p\right)}-i\right)\right\}_{i\geq 1}
\end{equation}
we obtain a random point process on $\left\{1,\ldots,p\right\}\times\Z$.
\begin{prop}\label{PropositionEffectiveSchurMeasure}
The probability of the point configuration (\ref{PoinConfiguration}) is determined by the probability measure
\begin{equation}\label{EffectiveSchurMeasure1}
\begin{split}
&\frac{1}{Z}s_{\lambda^{(\alpha_p)}}\left(1,q,\ldots,q^{B-\pi_1-1}\right)
s_{\lambda^{(\alpha_p)}/\lambda^{(\alpha_{p-1})}}\left(q^{A+\pi_1+2-\alpha_p},q^{A+\pi_1+3-\alpha_p},\ldots,q^{A+\pi_1+1-\alpha_{p-1}}\right)\\
&\times s_{\lambda^{(\alpha_{p-1})}/\lambda^{(\alpha_{p-2})}}\left(q^{A+\pi_1+2-\alpha_{p-1}},q^{A+\pi_1+3-\alpha_{p-1}},\ldots,q^{A+\pi_1+1-\alpha_{p-2}}\right)\\
&\times
\ldots
\cdot s_{\lambda^{(\alpha_2)}/\lambda^{(\alpha_1)}}\left(q^{A+\pi_1+2-\alpha_2},q^{A+\pi_1+3-\alpha_2},\ldots,q^{A+\pi_1+1-\alpha_{1}}\right)\\
&\times s_{\lambda^{(\alpha_1)}}\biggl(q^{A+\pi_1+2-\alpha_1},q^{A+\pi_1+3-\alpha_1},\ldots,q^{A+\pi_1+1-\beta_2};
q^{A+\pi_1+2-\beta_3},q^{A+\pi_1+3-\beta_3},\ldots,q^{A+\pi_1+1-\beta_4};\\
&\ldots;
q^{A+\pi_1+2-\beta_{2l-3}},q^{A+\pi_1+3-\beta_{2l-3}},\ldots,q^{A+\pi_1+1-\beta_{2l-2}};
q^{A+\pi_1+2-\beta_{2l-1}},q^{A+\pi_1+3-\beta_{2l-1}},\ldots,q^{A+\pi_1}
\biggr),
\end{split}
\end{equation}
where $Z$ is a normalization constant.
\end{prop}

Now we are ready to state the main result of the present work. Let $\Pi$ be a random skew partition with support $B^{A}/\pi$
whose weight is determined by equation (\ref{PartitionWeight}). Let $\left(\lambda^{1},\ldots,\lambda^{\left(A+B+1\right)}\right)$
be a sequence of Young diagrams associated with $\Pi$. Consider the subsequence
$\left(\lambda^{\left(\alpha_1\right)},\ldots,\lambda^{\left(\alpha_p\right)}\right)$ of
$\left(\lambda^{1},\ldots,\lambda^{\left(A+B+1\right)}\right)$, where the indexes $\alpha_1$, $\ldots$, $\alpha_p$ are chosen as it is described above. By assigning to this
subsequence the point configuration (\ref{PoinConfiguration})
we obtain a random point process on $\left\{1,\ldots,p\right\}\times\Z$. Set $n=B-\pi_1$, $q=e^{-\epsilon}$, and define
$$
x_j^k=e^{-\epsilon\lambda_j^{(\alpha_k)}},\; k=1,\ldots, p;\;j=1,\ldots,n.
$$
\begin{thm}\label{MainTheorem}As $\epsilon\rightarrow 0$, the point process formed by configurations
$$
\left\{\left(k,x_j^k\right)\vert k=1,\ldots,p;\;j=1,\ldots,n\right\}
$$
converges to the product matrix process
associated with truncated unitary matrices, described in Section \ref{SectionTruncatedProcess}, and defined by probability distribution (\ref{TruncatedProcessJD}). The parameters of the relevant product matrix process are given by
\begin{itemize}
  \item $n=B-\pi_1$.
  \item $m_k=A+B+1-\beta_{2k}$ for $1\leq k\leq l$, and $m_{l+k}=A+B+1-\alpha_k$ for $1\leq k\leq p-1$.
  \item $\nu_1=A+\pi_1+1-\alpha_1$, $\nu_k=A+\pi_1+1-\beta_{2k-1}$ for $2\leq k\leq l$, and $\nu_{l+k}=A+\pi_1+1-\alpha_{k+1}$ for $1\leq k\leq p-1$.
 \end{itemize}
\end{thm}

The truncated unitary matrices forming the product matrix process in Theorem \ref{MainTheorem} are shown schematically on Figure \ref{Fig_matrices}.

\begin{remark} The condition $m_1\geq 2n+\nu_1$ reads as
\begin{equation}
\alpha_1\geq B-\pi_1+\beta_2.
\end{equation}
The conditions $m_j\geq n+\nu_j+1$ (where $2\leq j\leq l$) can be rewritten as
$$
\beta_{2j-1}-\beta_{2j}\geq 0;\;\alpha_1-\beta_2\geq 0;
$$
and
$$
\alpha_{k+1}-\alpha_{k}\geq 0
$$
(where $1\leq k\leq p-1$), and are satisfied automatically.
\end{remark}

\begin{remark}
 The choice of the parameters $\nu_i$, $m_i$ is not unique. Namely, we need to identify the $l$ geometric
 series in  \eqref{varrho0specialization} with $l$ geometric series in the last two lines of \eqref{EffectiveSchurMeasure1} and there are $l!$ ways to do so. Formally, our proof goes through only for the choices which agree with the condition $m_1\ge 2 n+\nu_1$ of \eqref{Condition1}, see, however, Remark \ref{Remark_size_restriction}.
\end{remark}

\begin{example} Consider the particular case in which $\pi=\emptyset$, as in Figure \ref{Fig_Plane_partition}. In this situation $l=1$, $\beta_1=A+1$, $\beta_2=1$, and the parameters
$\alpha_1$, $\ldots$, $\alpha_p$  take values in $\left\{1,\ldots,A\right\}$. As a limit, we obtain the product matrix process with truncated unitary matrices whose parameters
$n$; $m_1$, $\ldots$, $m_p$; $\nu_1$, $\ldots$, $\nu_p$ are given by
\begin{itemize}
  \item $n=B$;
  \item $m_1=A+B$, and $m_{k}=A+B+1-\alpha_{k-1}$ for $2\leq k\leq p$;
  \item $\nu_k=A+1-\alpha_{k}$  for $1\leq k\leq p$.
 \end{itemize}
\end{example}

\begin{center}
\begin{figure}[h]
\setlength{\unitlength}{4pt}
\begin{picture}(80,150)
\put(0,0){\line(1,0){20}}
\put(0,0){\line(0,1){20}}
\put(0,20){\line(1,0){20}}
\put(20,0){\line(0,1){20}}
\put(40,0){\line(1,0){20}}
\put(40,0){\line(0,1){20}}
\put(40,20){\line(1,0){20}}
\put(60,0){\line(0,1){20}}
\put(25,10){\line(1,0){10}}
\put(35,10){\line(-1,1){2.5}}
\put(35,10){\line(-1,-1){2.5}}
\put(5,10){$U_{p+l-1}$}
\put(0,22){$A+B+1-\alpha_{p-1}$}
\put(-23,10){$A+B+1-\alpha_{p-1}$}
\put(45,10){$T_{p+l-1}$}
\put(40,22){$A+B+1-\alpha_{p-1}$}
\put(61,10){$A+B+1-\alpha_{p}$}
\put(10,25){$\vdots$}
\put(0,30){\line(1,0){20}}
\put(0,30){\line(0,1){20}}
\put(0,50){\line(1,0){20}}
\put(20,30){\line(0,1){20}}
\put(40,30){\line(1,0){20}}
\put(40,30){\line(0,1){20}}
\put(40,50){\line(1,0){20}}
\put(60,30){\line(0,1){20}}
\put(25,40){\line(1,0){10}}
\put(35,40){\line(-1,1){2.5}}
\put(35,40){\line(-1,-1){2.5}}
\put(5,40){$U_{l+1}$}
\put(0,52){$A+B+1-\alpha_{1}$}
\put(-23,40){$A+B+1-\alpha_{1}$}
\put(45,40){$T_{l+1}$}
\put(40,52){$A+B+1-\beta_{2l-1}$}
\put(61,40){$A+B+1-\alpha_{2}$}
\put(0,60){\line(1,0){20}}
\put(0,60){\line(0,1){20}}
\put(0,80){\line(1,0){20}}
\put(20,60){\line(0,1){20}}
\put(40,60){\line(1,0){20}}
\put(40,60){\line(0,1){20}}
\put(40,80){\line(1,0){20}}
\put(60,60){\line(0,1){20}}
\put(25,70){\line(1,0){10}}
\put(35,70){\line(-1,1){2.5}}
\put(35,70){\line(-1,-1){2.5}}
\put(5,70){$U_{l}$}
\put(0,82){$A+B+1-\beta_{2l}$}
\put(-23,70){$A+B+1-\beta_{2l}$}
\put(45,70){$T_{l}$}
\put(40,82){$A+B+1-\beta_{2l-3}$}
\put(61,70){$A+B+1-\beta_{2l-1}$}
\put(0,90){\line(1,0){20}}
\put(0,90){\line(0,1){20}}
\put(0,110){\line(1,0){20}}
\put(20,90){\line(0,1){20}}
\put(40,90){\line(1,0){20}}
\put(40,90){\line(0,1){20}}
\put(40,110){\line(1,0){20}}
\put(60,90){\line(0,1){20}}
\put(25,100){\line(1,0){10}}
\put(35,100){\line(-1,1){2.5}}
\put(35,100){\line(-1,-1){2.5}}
\put(5,100){$U_{2}$}
\put(0,112){$A+B+1-\beta_{4}$}
\put(-23,100){$A+B+1-\beta_{4}$}
\put(45,100){$T_{2}$}
\put(40,112){$A+B+1-\alpha_1$}
\put(61,100){$A+B+1-\beta_{3}$}
\put(0,120){\line(1,0){20}}
\put(0,120){\line(0,1){20}}
\put(0,140){\line(1,0){20}}
\put(20,120){\line(0,1){20}}
\put(40,120){\line(1,0){20}}
\put(40,120){\line(0,1){20}}
\put(40,140){\line(1,0){20}}
\put(60,120){\line(0,1){20}}
\put(25,130){\line(1,0){10}}
\put(35,130){\line(-1,1){2.5}}
\put(35,130){\line(-1,-1){2.5}}
\put(5,130){$U_{1}$}
\put(0,142){$A+B+1-\beta_{2}$}
\put(-23,130){$A+B+1-\beta_{2}$}
\put(45,130){$T_{1}$}
\put(40,142){$B-\pi_1$}
\put(61,130){$A+B+1-\alpha_1$}
\put(10,85){$\vdots$}
\end{picture}
\caption{The truncated unitary matrices forming the product matrix process associated with the random skew plane partitions. \label{Fig_matrices}}
\end{figure}
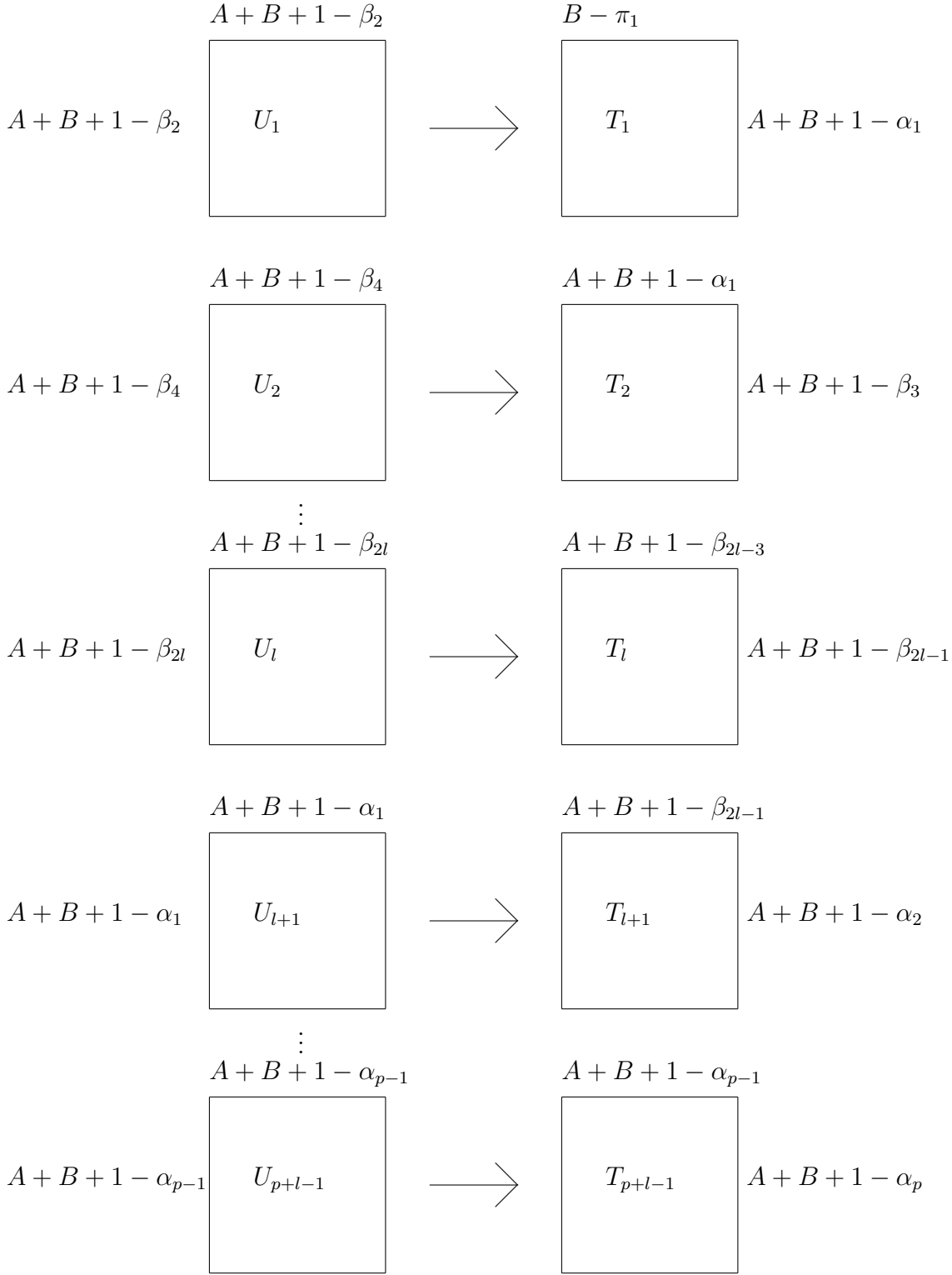
\end{center}
\clearpage

\section{Proof of Proposition \ref{PropositionProductTruncatedProcess}}
In order to prove  Proposition \ref{PropositionProductTruncatedProcess} we will use two results obtained in Kieburg, Kuijlaars, and Stivigny \cite{KieburgKuijlaarsStivigny}.
Namely, Corollary 2.6 in Kieburg, Kuijlaars, and Stivigny \cite{KieburgKuijlaarsStivigny} implies
that the probability distribution of $x^1=\left(x_1^1,\ldots,x_n^1\right)$ (which is the vector of the squared singular values of $T_l\ldots T_1$) can be written as
$$
\const \triangle\left(x^1\right)\det\left[w_k^{(l)}\left(x_j^1\right)\right]_{k,j=1}^ndx^1.
$$
The second result concerns the density of squared singular values for a product of a nonrandom and a truncated unitary matrix.
Namely, assume that $U$  is a Haar distributed unitary matrix of size $m\times m$, and let $T$ be an $(n+\nu)\times l$ truncation of $U$. In addition, let $X$ be a nonrandom matrix of size $l\times n$, and impose the following constraints for the parameters
$n$, $l$, $m$, and $\nu$:
$$
1\leq n\leq l\leq m,\;\;\; m\geq n+\nu+1.
$$
Denote by $\left(x_1,\ldots,x_n\right)$ the vector of squared singular values of $X$, and by $\left(y_1,\ldots,y_n\right)$ the vector of squared singular values of $TX$.
If $x_1,\ldots,x_n$ are pairwise distinct and nonzero, then the vector $\left(y_1,\ldots,y_n\right)$ has density
\begin{equation}
\const\left(\prod\limits_{j=1}^nx_j^{-m+n}\right)\left(\prod\limits_{j=1}^ny_j^{\nu}\right)
\det\left[\left(x_k-y_j\right)_+^{m-n-\nu-1}\right]_{j,k=1}^n
\frac{\triangle(y)}{\triangle(y)},
\end{equation}
see Kieburg, Kuijlaars, and Stivigny \cite{KieburgKuijlaarsStivigny}, Theorem 2.1. Now assume that $X=T_l\ldots T_1$. Applying the results stated above we immediately obtain
that the probability distribution of $\left(x_1^k,\ldots,x_n^k\right)$ is proportional to the product of determinants as in equation  (\ref{TruncatedProcessJD}).
In order to compute the normalization constant we use the Andr$\acute{\text{e}}$ief integral identity (see, for instance, De Bruijn \cite{Bruijn}), and the recurrence relation
$$
w_k^{(l+1)}(y)=\int\limits_0^1\tau^{\nu_{l+1}}\left(1-\tau\right)^{m_{l+1}-n-\nu_{l+1}-1}w_k^{(l)}\left(\frac{y}{\tau}\right)\frac{d\tau}{\tau},
$$
see Kieburg, Kuijlaars, and Stivigny \cite{KieburgKuijlaarsStivigny}, equation (2.22). The integration over $x_1^1$, $\ldots$, $x_n^1$ gives
\begin{equation}
\begin{split}
&\underset{0\leq x_1^1\leq\ldots\leq x_n^1<\infty}{\int\ldots\int}
\det\left[\left(x_j^2\right)^{\nu_{l+1}}\left(x_k^1-x_j^2\right)_+^{m_{l+1}-n-\nu_{l+1}-1}\left(x_k^1\right)^{n-m_{l+1}}\right]_{j,k=1}^n\det\left[w_k^{(l)}\left(x_j^1\right)\right]_{j,k=1}^ndx_1^1\ldots dx_n^1\\
&=\det\left[\int\limits_0^{\infty}
\left(x_j^2\right)^{\nu_{l+1}}\left(t-x_j^2\right)_+^{m_{l+1}-n-\nu_{l+1}-1}t^{n-m_{l+1}}w_k^{(l)}(t)dt\right]_{j,k=1}^n.
\end{split}
\nonumber
\end{equation}
Changing the integration variable $t=\frac{x_j^2}{\tau}$, we rewrite the integral inside the determinant above as
\begin{equation}
\begin{split}
&\left(x_j^2\right)^{\nu_{l+1}}\int\limits_0^1\left(\frac{x_j^2}{\tau}-x_j^2\right)^{m_{l+1}-n-\nu_{l+1}-1}
\frac{\left(x_j^2\right)^{n-m_{l+1}}}{\tau^{n-m_{l+1}}}w_k^{(l)}\left(\frac{x_j^2}{\tau}\right)\frac{x_j^2}{\tau^2}d\tau\\
&=\int\limits_0^1\tau^{\nu_{l+1}}\left(1-\tau\right)^{m_{l+1}-n-\nu_{l+1}-1}w_k^{(l)}\left(\frac{x_j^2}{\tau}\right)\frac{d\tau}{\tau}=
w_k^{(l+1)}\left(x_j^2\right).
\end{split}
\end{equation}

 As a result of integration over the variables $x_1^1$, $\ldots$, $x_n^1$; $\ldots$; $x_1^{p-1}$, $\ldots$, $x_n^{p-1}$ we find
$$
Z_{n,p+l}=\int\limits_{0\leq x_1\leq\ldots\leq x_n\leq\infty}\det\left(x_k^{j-1}\right)_{j,k=1}^n\det\left(w_k^{(l+p-1)}\left(x_j\right)\right)_{j,k=1}^n
dx_1\ldots dx_n.
$$
Applying the Andr$\acute{\text{e}}$ief integral identity again, we obtain
\begin{equation}
\begin{split}
&Z_{n,p+l}=\det\left[\int\limits_0^{\infty}x^{j-1}w_k^{(l+p-1)}\left(x\right)dx\right]_{j,k=1}^n\\
&=\left(c_{l+p-1}\right)^n\det\left[\int\limits_0^{\infty}x^{j-1}
G_{l+p-1,l+p-1}^{l+p-1,0}\left(\begin{array}{cccc}
                                     m_{l+p-1}-n, & \ldots, & m_2-n, & m_1-2n+k \\
                                     \nu_{l+p-1}, & \ldots, & \nu_2, & \nu_1+k-1
                                   \end{array}
\biggl|x\right)dx\right]_{j,k=1}^n.
\end{split}
\nonumber
\end{equation}
The integral inside the determinant can be computed explicitly in terms of Gamma functions.
Namely, formula (5.6.1.1) in Luke \cite{Luke} gives
\begin{equation}
\begin{split}
&\int\limits_0^{\infty}x^jG_{l+p-1,l+p-1}^{l+p-1,0}\left(\begin{array}{cccc}
                                                          m_{l+p-1} & \ldots & m_2-n & m_1-2n+k \\
                                                          \nu_{l+p-1} & \ldots & \nu_2 & \nu_1+k-1
                                                        \end{array}
\biggl| x\right)dx\\
&=\frac{\prod_{k=2}^{l+p-1}\Gamma\left(\nu_k+j\right)\Gamma\left(\nu_1+k-1+j\right)}{\prod_{k=2}^{l+p-1}\Gamma\left(m_k-n+j\right)
\Gamma\left(m_1-2n+k+j\right)},
\end{split}
\end{equation}
and we find
\begin{equation}
\begin{split}
&Z_{n,p+l}=\left[\Gamma\left(m_1-2n-\nu_1+1\right)\right]^n\prod\limits_{j=2}^{l+p-1}\left[\Gamma\left(m_j-n-\nu_j\right)\right]^n\\
&\times \frac{\prod_{k=2}^{l+p-1}\prod_{j=1}^n\Gamma\left(\nu_k+j\right)}{\prod_{k=2}^{l+p-1}\prod_{j=1}^n\Gamma\left(m_k-n+j\right)}
\det\left[\frac{\Gamma\left(\nu_1+k+j-1\right)}{\Gamma\left(m_1-2n+k+j\right)}\right]_{j,k=1}^n.
\end{split}
\nonumber
\end{equation}
The following formula  for the determinant  with Gamma functions entries is known
\begin{equation}
\det\left[\frac{\Gamma(c+i+j)}{\Gamma(d+i+j)}\right]_{i,j=0}^{n-1}=\prod\limits_{j=0}^{n-1}j!\frac{\Gamma(d-c+j)}{\Gamma(d-c)}\frac{\Gamma(c+j)}{\Gamma(d+n-1+j)},
\end{equation}
see equation (4.11) in Normand \cite{Normand}. Using this formula we obtain
$$
\det\left[\frac{\Gamma\left(\nu_1+k+j-1\right)}{\Gamma\left(m_1-2n+k+j\right)}\right]_{j,k=1}^n=\prod\limits_{j=1}^n\frac{\Gamma\left(\nu_1+j\right)}{\Gamma\left(m_1-n+j\right)}
\frac{\prod_{j=1}^n\Gamma(j)\Gamma\left(m_1-2n-\nu_1+j\right)}{\left(\Gamma\left(m_1-2n-\nu_1+1\right)\right)^n}.
$$
This gives
\begin{equation}
\begin{split}
&Z_{n,p+l}=\prod\limits_{j=2}^{l+p-1}\left[\Gamma\left(m_j-n-\nu_j\right)\right]^n\\
&\times \frac{\prod_{k=1}^{l+p-1}\prod_{j=1}^n\Gamma\left(\nu_k+j\right)}{\prod_{k=1}^{l+p-1}\prod_{j=1}^n\Gamma\left(m_k-n+j\right)}
\prod\limits_{j=1}^n\Gamma(j)\Gamma\left(m_1-2n-\nu_1+j\right).
\nonumber
\end{split}
\end{equation}
Since
$$
\prod\limits_{j=1}^n\frac{\Gamma\left(\nu+j\right)}{\Gamma(m-n+j)}=\prod\limits_{j=1}^{m-n-\nu}\frac{\Gamma(\nu+j)}{\Gamma(\nu+j+n)}
=\prod\limits_{j=1}^{m-n-\nu}\frac{1}{(j+\nu)_n},
$$
we can rewrite the normalization constant $Z_{n,p+l}$ in the same form as in the statement of the Proposition \ref{PropositionProductTruncatedProcess}.
\qed

\section{Limits of symmetric functions}\label{SectionLimitsSymmetricFunctions}
The aim of this Section is to obtain certain asymptotic formulae for the Schur functions, and for the skew Schur functions, see Proposition
\ref{PropositionLimitSchur} and Proposition \ref{LimitSkewSchur} below. We will need these formulae in the proofs of our main results (Proposition
\ref{PropositionGeneralLimit} and Theorem \ref{MainTheorem}).
\begin{prop}\label{PropositionLimitSchur} Let $\lambda$ be a Young diagram with $N$ rows, and assume that $M\geq N$. Set
$$
\lambda_1=-\frac{1}{\epsilon}\log r_1,\ldots,\lambda_N=-\frac{1}{\epsilon}\log r_N,
$$
where $0\leq r_1\leq\ldots\leq r_N<1$. Then
\begin{equation}\label{Equation41}
\begin{split}
&\underset{\epsilon\rightarrow 0+}{\lim}\left\{\epsilon^{MN-\frac{N(N+1)}{2}}s_{\lambda}\left(e^{-(1+\nu)\epsilon},\ldots,e^{-(M+\nu)\epsilon} \right)\right\}\\
&=\frac{1}{\prod\limits_{j=1}^N\Gamma(M-N+j)}\prod\limits_{i=1}^Nr_i^{1+\nu}\left(1-r_i\right)^{M-N}\prod\limits_{1\leq i<j\leq N}\left(r_j-r_i\right).
\end{split}
\end{equation}
The convergence is uniform in $r_j$'s.
\end{prop}
\begin{proof}
Homogeneity of the Schur polynomials implies
\begin{equation}\label{Schur1Limit}
s_{\lambda}\left(e^{-(1+\nu)\epsilon},\ldots,e^{-(M+\nu)\epsilon} \right)=e^{-(1+\nu)\epsilon\sum\limits_{i=1}^N\lambda_i}s_{\lambda}\left(1,e^{-\epsilon},\ldots,e^{-(M-1)\epsilon}\right).
\end{equation}
Moreover, the principal specialization of the Schur polynomials (see Macdonald \cite{Macdonald}, $\S$ 3, Example 1) gives
\begin{equation}\label{Zvezda}
s_{\lambda}\left(1,e^{-\epsilon},\ldots,e^{-(M-1)\epsilon}\right)=e^{-\epsilon\sum\limits_{i=1}^N(i-1)\lambda_i}\prod\limits_{1\leq i<j\leq M}
\frac{1-e^{-\epsilon\left(\lambda_i-\lambda_j-i+j\right)}}{1-e^{\epsilon\left(-i+j\right)}}.
\end{equation}
The product in the right-hand side of the expression  above can be written as
\begin{equation}
\begin{split}
&\prod\limits_{1\leq i<j\leq M}
\frac{1-e^{-\epsilon\left(\lambda_i-\lambda_j-i+j\right)}}{1-e^{\epsilon\left(-i+j\right)}}=\prod\limits_{1\leq i<j\leq N}\left(1-e^{-\epsilon\left(\lambda_i-\lambda_j
-i+j\right)}\right)\\
&\times\prod\limits_{i=1}^N\prod\limits_{j=N+1}^M\left(1-e^{-\epsilon\left(\lambda_i-i+j\right)}\right)
\underset{1\leq j\leq M}{\underset{1\leq i\leq N}{\prod\limits_{i<j}}}\frac{1}{1-e^{-\epsilon(-i+j)}}.
\end{split}
\end{equation}
Taking the limit $\epsilon\rightarrow 0+$, and using the fact that
\begin{equation}\label{Plimit}
\underset{1\leq j\leq M}{\underset{1\leq i\leq N}{\prod\limits_{i<j}}}\frac{1}{\epsilon(j-i)}=\frac{1}{\epsilon^{MN-\frac{N(N+1)}{2}}}\frac{1}{\prod\limits_{j=1}^N\Gamma(M-N+j)},
\end{equation}
we obtain the formula in the statement of the Proposition.
\end{proof}
\begin{lem}\label{PropositionHLimit} Let $k$, $m$ be two positive integers. Assume that $0<s,r<1$. If $k=-\frac{\log r}{\epsilon}$, and $m=-\frac{\log s}{\epsilon}$, then we have
\begin{equation}
\label{eq_h_convergence}
\underset{\epsilon\rightarrow 0+}{\lim}\left\{\epsilon^{M-1}h_{k-m}\left(e^{-(1+\nu)\epsilon},e^{-(2+\nu)\epsilon},\ldots,e^{-(M+\nu)\epsilon} \right)\right\}
=\frac{1}{\Gamma(M)}\frac{r^{1+\nu}}{s^{M+\nu}}\left(s-r\right)_+^{M-1},
\end{equation}
where $h_n=s_{(n)}$ is the $n$th complete homogeneous symmetric function, $M\geq 1$, $\nu>0$, and $(s-r)_+=\max\{0,s-r\}$.
The convergence is uniform in $s$, $r$.
\end{lem}
\begin{proof}
We start by noting that the right-hand side of \eqref{eq_h_convergence} has no singularity at $s=0$. Indeed, this follows from the observation that for positive $r$ we have $r^{1+\nu} (s-r)_+^{M-1}\le s^{1+\nu} s^{M-1}=s^{M+\nu}$.

Specializing (\ref{Zvezda}) to $\lambda=(m)$ gives
$$
h_m\left(1,q,\ldots,q^{M-1}\right)=\prod\limits_{j=1}^{M-1}\frac{1-q^{m+j}}{1-q^j}.
$$
Assuming  $m\leq k$ we obtain
\begin{equation}
\begin{split}
&h_{k-m}\left(e^{-(1+\nu)\epsilon},e^{-(2+\nu)\epsilon},\ldots,e^{-(M+\nu)\epsilon}\right)\\
&=e^{-(1+\nu)\epsilon(k-m)}h_{k-m}\left(1,e^{-\epsilon},\ldots,e^{-(M-1)\epsilon}\right)
=e^{-(1+\nu)\epsilon(k-m)}\prod\limits_{j=1}^{M-1}\frac{1-e^{-\epsilon(k-m+j)}}{1-e^{-\epsilon j}}\\
&\simeq\left(\frac{r}{s}\right)^{1+\nu}\prod\limits_{j=1}^{M-1}\frac{1-\frac{r}{s}}{\epsilon j}=\frac{1}{\epsilon^{M-1}}
\frac{1}{\Gamma(M)}\frac{r^{1+\nu}}{s^{M+\nu}}\left(s-r\right)^{M-1},
\end{split}
\end{equation}
where $s\geq r$. If $s<r$, then $m>k$, and $h_{k-m}\equiv 0$.
\end{proof}
\begin{prop}\label{LimitSkewSchur}
Let $\lambda$, $\mu$ be two Young diagrams, $l(\lambda)=l(\mu)=N$. Set
$$
\lambda_1=-\frac{1}{\epsilon}\log r_1,\ldots, \lambda_N=-\frac{1}{\epsilon}\log r_N,
$$
$$
\mu_1=-\frac{1}{\epsilon}\log s_1,\ldots, \mu_N=-\frac{1}{\epsilon}\log s_N,
$$
where $0<r_1\leq r_2\leq\ldots\leq r_N<1$, and $0<s_1\leq s_2\leq\ldots\leq s_N<1$. Then we have
\begin{equation}
\begin{split}
&\underset{\epsilon\rightarrow 0+}{\lim}\left\{\epsilon^{MN-N}s_{\lambda/\mu}\left(e^{-(1+\nu)\epsilon},e^{-(2+\nu)\epsilon},\ldots,e^{-(M+\nu)\epsilon} \right)\right\}\\
&=\frac{1}{\left(\Gamma(M)\right)^N}\frac{\prod_{i=1}^N\left(r_i\right)^{1+\nu}}{\prod_{i=1}^N\left(s_i\right)^{M+\nu}}
\det\left[\left(s_j-r_i\right)^{M-1}_+\right]_{i,j=1}^N.
\end{split}
\end{equation}
The convergence is uniform in $r_j$'s and $s_j$'s.
\end{prop}
\begin{proof} Apply the Jacobi-Trudi formula (see \cite[Chapter I, (5.4)]{Macdonald}),
$$
s_{\lambda/\mu}=\det\left(h_{\lambda_i-\mu_j-i+j}\right)_{i,j=1}^N,
$$
together with Lemma \ref{PropositionHLimit}.
\end{proof}
\begin{prop}\label{PropositionAsymptoticsSchurGeneral} Let $\lambda=\left(\lambda_1,\ldots,\lambda_n\right)$ be a Young diagram with $n$ rows. Set
$x_1=e^{-\epsilon\lambda_1}$, $\ldots$, $x_n=e^{-\epsilon\lambda_n}$. We have, uniformly for $x_1,\ldots,x_n\in[0,1]$,
\begin{equation}\label{AsymptoticSchurGeneral}
\begin{split}
&\underset{\epsilon\rightarrow 0+}{\lim}\left[\epsilon^{\sum\limits_{k=1}^p\left(m_k-n-\nu_k\right)n-\frac{n(n+1)}{2}}s_{\lambda}\left(e^{-(1+\nu_1)\epsilon},\ldots, e^{-(m_1-n)\epsilon};\ldots;
e^{-(1+\nu_p)\epsilon},\ldots, e^{-(m_p-n)\epsilon}\right)\right]\\
&=\frac{\left[\Gamma\left(m_1-2n-\nu_1+1\right)\right]^n}{\prod\limits_{j=1}^n\Gamma\left(m_1-2n-\nu_1+j\right)}\\
&\times
\det\left[G_{p,p}^{p,0}
\left(\begin{array}{cccc}
        m_p-n+1 & \ldots & m_2-n+1 & m_1-2n+k+1 \\
        \nu_p+1 & \ldots & \nu_{2}+1 & \nu_1+k
      \end{array}
\biggl| x_j\right)\right]_{j,k=1}^n,
\end{split}
\end{equation}
where the parameters $m_j$, $\nu_j$ are those specified in Section \ref{SectionTruncatedProcess}.
\end{prop}
\begin{proof} The proof is by induction over $p$. Assume that $p=1$. Then equation (\ref{AsymptoticSchurGeneral}) takes the form
\begin{equation}\label{P531}
\begin{split}
&\underset{\epsilon\rightarrow 0}{\lim}\left[\epsilon^{\left(m_1-n-\nu_1\right)n-\frac{n(n+1)}{2}}
s_{\lambda}\left(e^{-(1+\nu_1)\epsilon},\ldots,e^{-(m_1-n)\epsilon} \right)\right]\\
&=\frac{\left[\Gamma\left(m_1-2n-\nu_1+1\right)\right]^n}{\prod\limits_{j=1}^n\Gamma\left(m_1-2n-\nu_1+j\right)}
\det\left[G_{1,1}^{1,0}\left(\begin{array}{c}
                               m_1-2n+k+1 \\
                               \nu_1+k
                             \end{array}
\biggl|x_j\right)\right]_{j,k=1}^n.
\end{split}
\end{equation}
As follows from equations (2.20)-(2.25) in Kieburg, Kuijlaars, and Strivigny \cite{KieburgKuijlaarsStivigny},
$$
G_{1,1}^{1,0}\left(\begin{array}{c}
                               m_1-2n+k+1 \\
                               \nu_1+k
                             \end{array}
\biggl|x\right)=\frac{1}{\Gamma\left(m_1-2n-\nu_1+1\right)}x^{\nu_1+k}(1-x)^{m_1-2n-\nu_1}.
$$
We see that equation (\ref{P531}) turns into equation (\ref{Equation41}) (with $N=n$, $M=m_1-n-\nu$), and conclude that Proposition \ref{PropositionAsymptoticsSchurGeneral}
holds true for $p=1$.

Assume that the statement of Proposition \ref{PropositionAsymptoticsSchurGeneral} holds true for certain natural $p$. Let us prove this statement for $p+1$.
Equation (\ref{SummationFormula2}) implies
\begin{equation}\label{ConvolutionOfSchurFunctions}
\begin{split}
&s_{\lambda}\left(e^{-\left(1+\nu_1\right)\epsilon},\ldots, e^{-\left(m_1-n\right)\epsilon};\ldots;
e^{-\left(1+\nu_{p+1}\right)\epsilon},\ldots, e^{-\left(m_{p+1}-n\right)\epsilon}\right)\\
&=\sum\limits_{\mu}s_{\lambda/\mu}\left(e^{-\left(1+\nu_{p+1}\right)\epsilon},\ldots, e^{-\left(m_{p+1}-n\right)\epsilon}\right)
s_{\mu}\left(e^{-\left(1+\nu_1\right)\epsilon},\ldots, e^{-\left(m_1-n\right)\epsilon};\ldots,
e^{-\left(1+\nu_{p}\right)\epsilon},\ldots, e^{-\left(m_{p}-n\right)\epsilon}\right).
\end{split}
\end{equation}
Here we can assume that the sum is over Young diagrams $\mu$ with at most $n$ rows. Set
$$
\lambda_j=-\frac{1}{\epsilon}\log x_j, \; 1\leq j\leq n.
$$
If
$$
\mu_j=-\frac{1}{\epsilon}\log y_j, \; 1\leq j\leq n,
$$
then Proposition \ref{LimitSkewSchur} (with $N=n$, $M=m_{p+1}-n-\nu_{p+1}$, and $\nu=\nu_{p+1}$) implies that
$$
\epsilon^{-n+\left(m_{p+1}-n-\nu_{p+1}\right)n}s_{\lambda/\mu}\left(e^{-\left(1+\nu_{p+1}\right)\epsilon},\ldots, e^{-\left(m_{p+1}-n\right)\epsilon}\right)
$$
converges to
$$
\frac{1}{\left[\Gamma\left(m_{p+1}-n-\nu_{p+1}\right)\right]^n}\frac{\prod_{i=1}^n\left(x_i\right)^{1+\nu_{p+1}}}{\prod_{i=1}^n\left(y_i\right)^{m_{p+1}-n}}
\det\left[\left(y_j-x_i\right)_+^{m_{p+1}-n-\nu_{p+1}-1}\right]_{i,j=1}^n,
$$
and the convergence is uniform in $x_j$'s and $y_j$'s.  Moreover, by our assumption
$$
\epsilon^{\sum\limits_{k=1}^p\left(m_k-n-\nu_k\right)n-\frac{n(n+1)}{2}}s_{\mu}\left(e^{-\left(1+\nu_1\right)\epsilon},\ldots, e^{-\left(m_1-n\right)\epsilon};\ldots,
e^{-\left(1+\nu_{p}\right)\epsilon},\ldots, e^{-\left(m_{p}-n\right)\epsilon}\right)
$$
converges to
$$
\frac{\left[\Gamma\left(m_1-2n-\nu_1+1\right)\right]^n}{\prod\limits_{j=1}^n\Gamma\left(m_1-2n-\nu_1+j\right)}
\det\left[G_{p,p}^{p,0}
\left(\begin{array}{cccc}
        m_p-n+1 & \ldots & m_2-n+1 & m_1-2n+k+1 \\
        \nu_p+1 & \ldots & \nu_{2}+1 & \nu_1+k
      \end{array}
\biggl| y_j\right)\right]_{j,k=1}^n,
$$
uniformly for $y_1,\ldots, y_n\in[0,1]$. Since $d\mu_j=-\frac{dy_j}{\epsilon y_j}$, we conclude that the right-hand side
of equation (\ref{ConvolutionOfSchurFunctions}) multiplied by
$$
\epsilon^{\sum\limits_{k=1}^{p+1}\left(m_k-n-\nu_k\right)n-\frac{n(n+1)}{2}}
$$
converges to
\begin{equation}
\begin{split}
&\frac{\left[\Gamma\left(m_1-2n-\nu_1+1\right)\right]^n}{\prod\limits_{j=1}^n\Gamma\left(m_1-2n-\nu_1+j\right)}
\frac{1}{\left[\Gamma\left(m_{p+1}-n-\nu_{p+1}\right)\right]^n}
\prod\limits_{i=1}^n\left(x_i\right)^{1+\nu_{p+1}}\\
&\times\underset{0\leq y_1\leq\ldots\leq y_n\leq 1}{\int\ldots\int}
\prod\limits_{i=1}^n\left(y_i\right)^{n-m_{p+1}}\det\left[\left(y_j-x_i\right)_+^{m_{p+1}-n-\nu_{p+1}-1}\right]_{i,j=1}^n\\
&\times\det\left[G_{p,p}^{p,0}
\left(\begin{array}{cccc}
        m_p-n+1 & \ldots & m_2-n+1 & m_1-2n+k+1 \\
        \nu_p+1 & \ldots & \nu_{p-1}+1 & \nu_1+k
      \end{array}
\biggl| y_j\right)\right]_{j,k=1}^n\frac{dy_1}{y_1}\ldots\frac{dy_n}{y_n},
\end{split}
\nonumber
\end{equation}
and the convergence is uniform in $x_1,\ldots,x_n$. The expression above can be rewritten as
\begin{equation}\label{EABOVE}
\begin{split}
&\frac{1}{\prod\limits_{j=1}^n\Gamma\left(m_1-2n-\nu_1+j\right)\prod\limits_{j=2}^{p+1}\left[\Gamma\left(m_j-n-\nu_j\right)\right]^n}\\
&\times\underset{0\leq y_1\leq\ldots\leq y_n\leq 1}{\int\ldots\int}dy_1\ldots dy_n
\prod\limits_{i=1}^n\left(y_i\right)^{n-m_{p+1}-1}\det\left[x_k^{1+\nu_p}\left(y_j-x_k\right)_+^{m_{p+1}-n-\nu_{p+1}-1}\right]_{k,j=1}^n\det\left[w_k^{(p)}\left(y_j\right)\right]_{j,k=1}^n,
\end{split}
\nonumber
\end{equation}
where we have used equations (\ref{TheWeightFunctions}), (\ref{cl}) to rewrite the Meijer $G$-function in terms of the corresponding weight function $w_{k}^{(p)}(y)$.
By the Andr$\acute{\mbox{e}}$ief  identity, and by the same calculations as in the proof of Proposition \ref{PropositionProductTruncatedProcess} it can be shown that
(\ref{EABOVE}) is equal to
$$
\frac{1}{\prod\limits_{j=1}^n\Gamma\left(m_1-2n-\nu_1+j\right)
\prod\limits_{j=2}^{p+1}\left[\Gamma\left(m_j-n-\nu_j\right)\right]^n}\det\left[x_jw_k^{(p+1)}\left(x_j\right)\right]_{j,k=1}^n.
$$
Rewriting the weight function $w_k^{(p+1)}$ in terms of the corresponding Meijer $G$-function we obtain that
\begin{equation}
\begin{split}
&\underset{\epsilon\rightarrow 0+}{\lim}\left[\epsilon^{\sum\limits_{k=1}^{p+1}\left(m_k-n-\nu_k\right)n-\frac{n(n+1)}{2}}s_{\lambda}\left(e^{-(1+\nu_1)\epsilon},\ldots, e^{-(m_1-n)\epsilon};\ldots;
e^{-(1+\nu_{p+1})\epsilon},\ldots, e^{-(m_{p+1}-n)\epsilon}\right)\right]\\
&=\frac{\left[\Gamma\left(m_1-2n-\nu_1+1\right)\right]^n}{\prod\limits_{j=1}^n\Gamma\left(m_1-2n-\nu_1+j\right)}\\
&\times
\det\left[G_{p+1,p+1}^{p+1,0}
\left(\begin{array}{cccc}
        m_{p+1}-n+1 & \ldots & m_2-n+1 & m_1-2n+k+1 \\
        \nu_{p+1}+1 & \ldots & \nu_{2}+1 & \nu_1+k
      \end{array}
\biggl| x_j\right)\right]_{j,k=1}^n,
\end{split}
\end{equation}
uniformly for $x_1,\ldots,x_n\in[0,1]$.
\end{proof}
\section{Convergence of the Schur process to the product matrix process with truncated unitary matrices. Proof of Proposition \ref{PropositionGeneralLimit}}
Now we begin to investigate the convergence of the Schur process to the product matrix process with truncated unitary matrices.
We start with the case where the initial conditions are defined by a single truncated matrix $T_1$, see Proposition \ref{TheoremConvergenceSchurProcess}
below. Then (using Proposition  \ref{PropositionAsymptoticsSchurGeneral}) we generalize
Proposition \ref{TheoremConvergenceSchurProcess} to the case where the initial conditions are specified by a product of $l$ truncated matrices $T_l\ldots T_1$, and prove
Proposition \ref{PropositionGeneralLimit}. We remark that, in principle, the second part can be avoided, as the general $\ell$ case can be obtained from the $\ell=1$ case by restriction of a distribution to a subset of matrices. In particular, in this way the consistency of Proposition \ref{PropositionProductTruncatedProcess} between different values of $\ell$ can be used to produce an alternative proof of Proposition \ref{PropositionAsymptoticsSchurGeneral}.

\medskip

Consider the Schur process defined in Section \ref{TheSchurProcesses}. Define the specializations
$\varrho_0^+$, $\ldots$, $\varrho_{p-1}^+$, $\varrho_1^{-}$, $\ldots$, $\varrho_{p}^{-}$ of the algebra of symmetric functions as follows
\begin{itemize}
  \item The specialization $\varrho_{p}^{-}$ is defined by
  $$
  \varrho_{p}^{-}=\left(1,e^{-\epsilon},\ldots, e^{-(n-1)\epsilon}\right).
  $$
  All other $\varrho_1^{-}$, $\ldots$, $\varrho_{p-1}^{-}$ are trivial.
  \item The  specializations $\varrho_0^+$, $\ldots$, $\varrho_{p-1}^+$ are defined by
  $$
  \varrho_{j}^+=\left(e^{-\left(1+\nu_j\right)\epsilon},e^{-\left(2+\nu_j\right)\epsilon},\ldots,e^{-\left(m_j-n\right)\epsilon}\right),\; 0\leq j\leq p-1.
  $$
\end{itemize}
With these specializations the Schur process lives on the point configurations
\begin{equation}
\left\{\left(1,\lambda_i^{(1)}-i\right)\right\}_{i=1}^{n}\cup\ldots\cup\left\{\left(1,\lambda_i^{(p)}-i\right)\right\}_{i=1}^{n}.
\end{equation}
Set
\begin{equation}\label{rconfigurations}
\begin{split}
x_j^k=e^{-\epsilon\lambda_j^{(k)}},\; k=1,\ldots,p;\; j=1,\ldots,n.
\end{split}
\end{equation}
The above Schur process induces a point process on $\left\{1,\ldots,n\right\}\times\R_{>0}$, and this process
is formed by  the configurations
\begin{equation}\label{rconfigurations1}
\left\{\left(k,x_j^k\right)\biggl|k=1,\ldots,p;\; j=1,\ldots,n\right\}.
\end{equation}
\begin{prop}\label{TheoremConvergenceSchurProcess}
As $\epsilon\rightarrow 0$, the point process formed by configurations (\ref{rconfigurations1}) converges to the product matrix process
associated with truncated unitary matrices, and defined in Section \ref{SectionTruncatedProcess}. The initial conditions of this process are defined by the matrix
$T_1$ (which is the truncation of $U_1$).
\end{prop}
\begin{proof}Since  specializations $\varrho_1^-$, $\ldots$, $\varrho_{p-1}^-$ are trivial, the Schur
process turns into the probability measure
\begin{equation}\label{EffZ}
\begin{split}
&\Prob\left\{\lambda^{(1)},\lambda^{(2)},\ldots,\lambda^{(p-1)},\lambda^{(p)}\right\}\\
&=\frac{1}{Z}s_{\lambda^{(1)}}\left(\varrho_0^{+}\right)s_{\lambda^{(2)}/\lambda^{(1)}}\left(\varrho_1^+\right)
s_{\lambda^{(3)}/\lambda^{(2)}}\left(\varrho_2^+\right)\ldots s_{\lambda^{(p-1)}/\lambda^{(p-2)}}\left(\varrho_{p-2}^+\right)s_{\lambda^{(p)}/\lambda^{(p-1)}}\left(\varrho_{p-1}^+\right)
s_{\lambda^{(p)}}\left(\varrho_p^-\right),
\end{split}
\end{equation}
where
$$
\frac{1}{Z}=\frac{1}{H\left(\varrho_0^+;\varrho_p^-\right)H\left(\varrho_1^+;\varrho_p^-\right)\ldots H\left(\varrho_{p-1}^+;\varrho_p^-\right)}.
$$
Now we use the asymptotic formulae for the Schur functions obtained in Section \ref{SectionLimitsSymmetricFunctions}. In particular, Proposition \ref{PropositionLimitSchur}
gives
\begin{equation}
\begin{split}
&s_{\lambda^{(1)}}\left(e^{-\left(1+\nu_1\right)\epsilon},e^{-\left(2+\nu_1\right)\epsilon},\ldots,e^{-\left(m_1-n\right)\epsilon}\right)\\
&\simeq\frac{\epsilon^{-\left(m_1-n-\nu_1\right)n+\frac{n(n+1)}{2}}}{\prod\limits_{j=1}^n\Gamma\left(m_1-2n-\nu_1+j\right)}
\prod\limits_{i=1}^n\left(x_i^1\right)^{1+\nu_1}\left(1-x_i^1\right)^{m_1-2n-\nu_1}\prod\limits_{1\leq i<j\leq n}\left(x_j^1-x_i^1\right),
\end{split}
\end{equation}
as $\epsilon\rightarrow 0+$. In addition, Proposition \ref{LimitSkewSchur} implies that as $\epsilon\rightarrow 0+$,
\begin{equation}
\begin{split}
&s_{\lambda^{(j)}/\lambda^{(j-1)}}\left(e^{-\left(1+\nu_j\right)\epsilon},e^{-\left(2+\nu_j\right)\epsilon},\ldots,e^{-\left(m_j-n\right)\epsilon}\right)\\
&\simeq\frac{\epsilon^{-\left(m_j-n-\nu_j\right)n+n}}{\left[\Gamma\left(m_j-n-\nu_j\right)\right]^n}
\prod\limits_{i=1}^n\frac{\left(x_i^j\right)^{1+\nu_j}}{\left(x_i^{j-1}\right)^{m_j-n}}\det\left[\left(x_k^{j-1}-x_l^j\right)_+^{m_j-n-\nu_j-1}\right]_{k,l=1}^n,
\end{split}
\end{equation}
where $2\leq j\leq p$. Besides,  Proposition \ref{PropositionLimitSchur} also gives
\begin{equation}\label{StrivialA}
s_{\lambda^{(p)}}\left(1,e^{-\epsilon},\ldots,e^{-(n-1)\epsilon}\right)\simeq\frac{1}{\epsilon^{\frac{n(n-1)}{2}}\prod\limits_{j=1}^n\Gamma(j)}\prod\limits_{1\leq i<j\leq n}
\left(x_j^p-x_i^p\right).
\end{equation}
Let us find the asymptotics of the normalization constant
(defined by equation (\ref{EffZ}))\footnote{This is not strictly necessary as the uniform convergence of the configuration weights implies the convergence of the normalization constants. We perform this limit transition for the sake of completeness.}. We have
$$
\frac{1}{H\left(\varrho_{j-1}^+;\varrho_p^-\right)}=\prod\limits_{i=0}^{n-1}\prod\limits_{l=1}^{m_j-n-\nu_j}
\left(1-e^{-\left(i+l+\nu_j\right)\epsilon}\right)\simeq\epsilon^{\left(m_j-n-\nu_j\right)n}\prod\limits_{i=0}^{n-1}\prod\limits_{l=1}^{m_j-n-\nu_j}
\left(i+l+\nu_j\right),
$$
where $1\leq j\leq p$. Therefore,
\begin{equation}\label{ConstantAsymptotics}
\begin{split}
&\frac{1}{Z}\simeq\epsilon^{\sum\limits_{j=1}^p\left(m_j-n-\nu_j\right)n}\prod\limits_{i=0}^{n-1}\prod\limits_{j=1}^p\prod\limits_{l=1}^{m_j-n-\nu_j}
\left(i+l+\nu_j\right)\\
&=\epsilon^{\sum\limits_{j=1}^p\left(m_j-n-\nu_j\right)n}\prod\limits_{j=1}^p\prod\limits_{l=1}^{m_j-n-\nu_j}
\frac{\Gamma\left(l+\nu_j+n\right)}{\Gamma\left(l+\nu_j\right)}.
\end{split}
\end{equation}
Taking into account that
$$
d\lambda^{(j)}_i\sim-\frac{dx_i^j}{\epsilon x_i^j},\;\;\; 1\leq j\leq p,\;\;\; 1\leq i\leq n,
$$
we find that the probability measure on $\lambda^{(1)}$, $\lambda^{(2)}$,$\ldots$, $\lambda^{(p)}$ turns into the probability measure
\begin{equation}
\begin{split}
&\left(\prod\limits_{j=1}^p\prod\limits_{l=1}^{m_j-n-\nu_j}\frac{\Gamma\left(l+\nu_j+n\right)}{\Gamma\left(l+\nu_j\right)}\right)
\frac{1}{\prod\limits_{j=1}^n\Gamma(j)\Gamma\left(m_1-2n-\nu_1+j\right)}
\frac{1}{\prod\limits_{j=2}^p\left[\Gamma\left(m_j-n-\nu_j\right)\right]^n}\\
&\times\triangle\left(x^p\right)\prod\limits_{j=2}^p\det\left[\left(x_l^j\right)^{\nu_j}\left(x_k^{j-1}-x_l^j\right)^{m_j-n-\nu_j-1}
x_k^{n-m_j}\right]_{k,l=1}^n\\
&\times\det\left[\left(x_l^1\right)^{\nu_1+k-1}\left(1-x_l^1\right)^{m_1-2n-\nu_1}\right]_{l,k=1}^n
dx^1\ldots dx^n.
\end{split}
\end{equation}
This probability measure can be interpreted as the product matrix process associated with truncated unitary matrices, see Proposition \ref{PropositionProductTruncatedProcess}.
The initial conditions of this process are defined by  the single truncated unitary matrix $T_1$, which corresponds to $l=1$ in the definition of this product matrix process in Section
\ref{SectionTruncatedProcess}.
\end{proof}
\begin{prop}\label{PropositionConvergencetoSingularValues}Consider the Schur measure
\begin{equation}\label{S1}
\frac{1}{Z}s_{\lambda^{(p)}}\left(\varrho_0^+,\varrho_1^+,\ldots,\varrho_{p-1}^+\right)s_{\lambda^{(p)}}\left(\varrho_p^-\right),
\end{equation}
which is the projection\footnote{For a discussion of projections of Schur processes see Borodin \cite{Borodin}, Section 2.} of the Schur process (equation (\ref{EffZ})) to the Young diagram $\lambda^{(p)}$. If
$$
x_k^p=e^{-\epsilon\lambda^{(p)}_k},\;\;\; 1\leq k\leq n,
$$
then as $\epsilon\rightarrow 0+$ the  probability measure defined by equation (\ref{S1}) converges to
\begin{equation}\label{EnsembleSingularValues}
\begin{split}
&\frac{\prod_{k=1}^p\prod_{j_k=1}^{m_k-n-\nu_k}\left(j_k+\nu_k\right)_n}{\prod_{j=1}^n\Gamma\left(m_1-2n-\nu_1+j\right)\Gamma(j)\prod_{k=2}^p
\left(\Gamma\left(m_k-n-\nu_k\right)\right)^n}\\
&\times\triangle\left(x^p\right)\det\left[w_k^{(p)}\left(x_j\right)\right]_{k,j=1}^ndx_1^p\ldots dx_n^p.
\end{split}
\end{equation}
The probability measure defined by expression (\ref{EnsembleSingularValues}) can be understood as the ensemble
of squared singular values of the product matrix $T_p\ldots T_1$ (where $T_1$, $\ldots$, $T_p$ are the truncated unitary matrices defined in Section
\ref{SectionTruncatedProcess}).
\end{prop}
\begin{proof} We already know that
the Schur process defined by equation (\ref{EffZ})  converges to the product matrix process defined by
the probability distribution (\ref{TruncatedProcessJD}). The Schur measure defined by expression (\ref{S1})
is the projection of this Schur process to the Young diagram $\lambda^{(p)}$, and the probability measure defined by expression
(\ref{EnsembleSingularValues}) can be understood as the projection of the product matrix process defined by
the probability distribution (\ref{TruncatedProcessJD}) to $x^p=\left(x_1^p,\ldots,x_n^p\right)$.
The result follows.
\end{proof}

\textit{Proof of Proposition} \ref{PropositionGeneralLimit}.\\
Proposition \ref{PropositionGeneralLimit} is a generalization of Proposition \ref{TheoremConvergenceSchurProcess} to the situation where the initial conditions are defined by a product of $l$ truncated matrices (and not by a single truncated matrix).
For the specializations specified in the statement of Proposition \ref{PropositionGeneralLimit} the Schur process turns into the probability measure
\begin{equation}\label{EffectiveSchurMeasure}
\begin{split}
&\frac{1}{Z}s_{\lambda^{(p)}}\left(1,e^{-\epsilon},\ldots,e^{-(n-1)\epsilon}\right)
s_{\lambda^{(p)}/\lambda^{(p-1)}}\left(e^{-\left(1+\nu_{l+p-1}\right)\epsilon},
e^{-\left(2+\nu_{l+p-1}\right)\epsilon},\ldots,e^{-\left(m_{l+p-1}-n\right)\epsilon}\right)\\
&s_{\lambda^{(p-1)}/\lambda^{(p-2)}}\left(e^{-\left(1+\nu_{l+p-2}\right)\epsilon},
e^{-\left(2+\nu_{l+p-2}\right)\epsilon},\ldots,e^{-\left(m_{l+p-2}-n\right)\epsilon}\right)\\
&\vdots\\
&s_{\lambda^{(2)}/\lambda^{(1)}}\left(e^{-\left(1+\nu_{l+1}\right)\epsilon},
e^{-\left(2+\nu_{l+1}\right)\epsilon},\ldots,e^{-\left(m_{l+1}-n\right)\epsilon}\right)\\
&s_{\lambda^{(1)}}\left(e^{-\left(1+\nu_{1}\right)\epsilon},
e^{-\left(2+\nu_{1}\right)\epsilon},\ldots,e^{-\left(m_{1}-n\right)\epsilon};\ldots;
e^{-\left(1+\nu_{l}\right)\epsilon},e^{-\left(2+\nu_{l}\right)\epsilon},\ldots,e^{-\left(m_{l}-n\right)\epsilon}\right),
\end{split}
\end{equation}
where $Z$ is a normalization constant. The rest of the proof of Proposition
\ref{PropositionGeneralLimit} follows the same line as that of
Proposition \ref{TheoremConvergenceSchurProcess}. The only essential difference is that  we use equation (\ref{AsymptoticSchurGeneral}) instead
of equation (\ref{EffZ}) for the asymptotics of the relevant Schur function
$$
s_{\lambda}\left(e^{-(1+\nu_1)\epsilon},\ldots, e^{-(m_1-n)\epsilon};\ldots;
e^{-(1+\nu_l)\epsilon},\ldots, e^{-(m_l-n)\epsilon}\right)
$$
as $\epsilon\rightarrow 0+$.
\qed
\section{Convergence of the correlation functions and the proof of Proposition \ref{THEOREMCorrelationKernel}}\label{Section 6}
Consider the Schur process defined by  probability measure (\ref{SchurProcess}). Assume that the  specializations
$\varrho_0^+$, $\ldots$, $\varrho_{p-1}^+$ of the Schur process are defined by equations (\ref{varrho0specialization})-
(\ref{varrhopminus1specialization}), and that the  specialization $\varrho_{p}^{-}$ is defined by
equation (\ref{negativespecialization}).
Let us agree that all other  $\varrho_1^{-}$, $\ldots$, $\varrho_{p-1}^{-}$ are trivial.
With these specializations the Schur process can be understood as a point process on $\left\{1,\ldots,p\right\}\times\Z$ formed by point configurations
(\ref{CSP}).
Denote by $\varrho_{k_1,\ldots,k_m}^{\epsilon,\Schur}\left(u_1^1,\ldots,u_{k_1}^1;\ldots;u_1^m,\ldots,u_{k_m}^m\right)$
the correlation functions of this Schur process, and by $K^{\epsilon}_{\Schur}\left(r,u;s,v\right)$ the correlation kernel of this Schur process.

Proposition \ref{PropositionGeneralLimit} says that the Schur process under consideration converges to the product matrix process with truncated unitary matrices. This implies  convergence of the correlation functions. Namely, if
$\varrho_{k_1,\ldots,k_m}\left(x_1^1,\ldots,x_{k_1}^1;\ldots;x_1^m,\ldots,x_{k_m}^m\right)$ denotes the correlation function of the product matrix process with truncated unitary matrices defined by probability measure (\ref{TruncatedProcessJD}), then we must have
\begin{equation}
\begin{split}
&\underset{\epsilon\rightarrow 0+}{\lim}\biggl\{
\frac{1}{\prod\limits_{i=1}^{k_1}\epsilon x_i^1\ldots\prod\limits_{i=1}^{k_m}\epsilon x_i^m}
\varrho_{k_1,\ldots,k_m}^{\epsilon,\Schur}\left(-\frac{1}{\epsilon}\log x_1^1,\ldots,-\frac{1}{\epsilon}\log x_{k_1}^1;\ldots;-\frac{1}{\epsilon}\log x_1^m,\ldots,-\frac{1}{\epsilon}\log x_{k_m}^m
\right)\biggr\}\\
&=\varrho_{k_1,\ldots,k_m}\left(x_1^1,\ldots,x_{k_1}^1;\ldots;x_1^m,\ldots,x_{k_m}^m\right),
\end{split}
\nonumber
\end{equation}
where the denominator came from the coordinate change.
This limiting relation between the correlation functions would naturally follow from the limiting relation between the correlation kernels,
\begin{equation}\label{Conv1}
\underset{\epsilon\rightarrow 0+}{\lim}
\left\{\frac{1}{\epsilon y}\hat{K}^{\epsilon}_{\Schur}\left(r,-\frac{1}{\epsilon}\log x;s,-\frac{1}{\epsilon}\log y\right)\right\}
=K_{n,p,l}\left(r,x;s,y\right),
\end{equation}
where $K_{n,p,l}\left(r,x;s,y\right)$ denotes the correlation kernel of the product matrix process with truncated unitary matrices, and where
$\hat{K}^{\epsilon}_{\Schur}\left(r,u;s,v\right)$ stands for a kernel equivalent to $K^{\epsilon}_{\Schur}\left(r,u;s,v\right)$.
\footnote{Recall that two kernels of a determinantal process are called equivalent if they define the same correlation functions.
In what follows we will choose $\hat{K}^{\epsilon}_{\Schur}\left(r,u;s,v\right)$ in such a way that the limit in the righthand side of equation
(\ref{Conv1}) will exist.}

The Okounkov-Reshetikhin formula for the correlation kernel $K_{\Schur}^{\epsilon}\left(r,u;s,v\right)$ is
\begin{equation}\label{CorrelationKernelSchurProcess}
\begin{split}
K_{\Schur}^{\epsilon}\left(r,u;s,v\right)=\frac{1}{\left(2\pi i\right)^2}
\oint\limits_{\Sigma_z}\oint\limits_{\Sigma_w}\frac{H\left(\varrho_{[r,p]}^-;z\right)H\left(\varrho_{[0,s)}^+;w\right)}{(zw-1)H\left(\varrho_{[0,r)}^+;z^{-1}\right)
H\left(\varrho_{[s,p]}^-;w^{-1}\right)}\frac{dzdw}{z^{u+1}w^{v+1}},
\end{split}
\end{equation}
see Theorem 2.2 in Ref. \cite{BorodinRains}.
The choice of the integration contours $\Sigma_z$ and $\Sigma_{w}$ depends on the time variables $r$ and $s$, and will be specified below.
In the formula for the correlation kernel $K_{\Schur}^{\epsilon}\left(r,u;s,v\right)$  above
we have used the notation
$$
\varrho_{[i,j]}^{\pm}=\varrho_i^{\pm}\cup\varrho_{i+1}^{\pm}\cup\ldots\cup\varrho_j^{\pm}.
$$
The specializations
$\varrho_0^+$, $\ldots$, $\varrho_{p-1}^+$, $\varrho_1^{-}$, $\ldots$, $\varrho_{p}^{-}$ of the algebra of symmetric functions
are specified by equations (\ref{varrho0specialization})-(\ref{negativespecialization}).

For simplicity, let us consider the case corresponding to $l=1$. The proof of Proposition \ref{THEOREMCorrelationKernel}
for a general $l$ is essentially  the same.

We find
$$
\frac{1}{H\left(\varrho_{[s,p]}^-;w^{-1}\right)}=
\prod\limits_{a_0=1}^n\left(1-e^{-(a_0-1)\epsilon}w^{-1}\right),
\;
H\left(\varrho_{[r,p]}^-;z\right)=
\frac{1}{\prod\limits_{a_0=1}^n\left(1-e^{-(a_0-1)\epsilon}z\right)},
$$
$$
H\left(\varrho_{[0,s)}^+;w\right)=\frac{1}{\prod_{b_1=1+\nu_1}^{m_1-n}\left(1-e^{-b_1\epsilon}w\right)
\ldots\prod_{b_s=1+\nu_s}^{m_s-n}\left(1-e^{-b_s\epsilon}w\right)},
$$
and
$$
\frac{1}{H\left(\varrho_{[0,r)}^+;w\right)}=\prod_{a_1=1+\nu_1}^{m_1-n}\left(1-e^{-a_1\epsilon}z^{-1}\right)
\ldots\prod_{a_r=1+\nu_r}^{m_r-n}\left(1-e^{-a_r\epsilon}z^{-1}\right).
$$
Therefore, we can write
\begin{equation}\label{K1}
\begin{split}
&K_{\Schur}^{\epsilon}\left(r,u;s,v\right)=\frac{1}{\left(2\pi i\right)^2}
\oint\limits_{\Sigma_z}\oint\limits_{\Sigma_w}\frac{dzdw}{(zw-1)z^{u+1}w^{v+1}}\prod_{a_0=1}^n\frac{1-e^{-\left(a_0-1\right)\epsilon}w^{-1}}{1-e^{-\left(a_0-1\right)\epsilon}z}\\
&\times\frac{\prod_{a_1=1+\nu_1}^{m_1-n}\left(1-e^{-a_1\epsilon}z^{-1}\right)
\ldots\prod_{a_r=1+\nu_r}^{m_r-n}\left(1-e^{-a_r\epsilon}z^{-1}\right)}{\prod_{b_1=1+\nu_1}^{m_1-n}\left(1-e^{-b_1\epsilon}w\right)
\ldots\prod_{b_s=1+\nu_s}^{m_s-n}\left(1-e^{-b_s\epsilon}w\right)}.
\end{split}
\end{equation}
Assume that $r\geq s$.  According to Theorem 2.2 in Ref. \cite{BorodinRains}, we can choose $\Sigma_{w}$ as a counterclockwise circle contour with its center at $0$, whose radius is larger than 1. Moreover,
$\Sigma_{w}$  should be chosen in such a way  that all the points $e^{(1+\nu_1)\epsilon}$, $\ldots$, $e^{(m_1-n)\epsilon}$; $\ldots$; $e^{(1+\nu_s)\epsilon}$, $\ldots$, $e^{(m_s-n)\epsilon}$
of the complex $w$-plane will be  situated outside of the circle $\Sigma_{w}$ on its right. So we define the contour $\Sigma_{w}$ by
$$
\Sigma_w=\left\{w:w=e^{\epsilon\beta+i\varphi},\varphi\in[-\pi,\pi)\right\},
$$
where $0<\beta<\min{\left\{\left(1+\nu_1\right),\ldots,\left(m_1-n\right);\ldots;\left(1+\nu_s\right),\ldots,\left(m_s-n\right)\right\}}$.

 In addition, we can choose $\Sigma_z$  as a counterclockwise circle with its center at $0$, such that $|z||w|>1$ for all $z\in\Sigma_z$, for all $w\in\Sigma_w$, and such that the points $1$, $e^{\epsilon}$, $\ldots$, $e^{(n-1)\epsilon}$ of the complex $z$-plane are situated outside of $\Sigma_z$ . So we define $\Sigma_z$ by
 $$
 \Sigma_z=\left\{z: z=e^{-\alpha\epsilon+i\varphi}, \varphi\in[-\pi,\pi)\right\},
 $$
 where $0<\alpha<\beta$.

Let us consider the integral over $w$ in  (\ref{K1}). Since $v\geq-n$, the residue at infinity of the integrand is equal to zero.
Moreover, since $\beta>\alpha>0$, the singular point $w=\frac{1}{z}$ is situated inside the contour $\Sigma_w$, for every $z\in\Sigma_z$.
This enables us  (without changing the integral in the right-hand side of equation (\ref{K1})) transform
$\Sigma_w$ through the extended $w$-complex plane into an integration contour $\Sigma_w'$ which encircles all the points
$e^{(1+\nu_1)\epsilon}$, $\ldots$, $e^{(m_1-n)\epsilon}$; $\ldots$; $e^{(1+\nu_s)\epsilon}$, $\ldots$, $e^{(m_s-n)\epsilon}$ once in the clockwise direction, and leaves $w=e^{\epsilon\beta}$ on its left.
The contour $\Sigma_w'$ can be viewed as an image of a contour $C_{\zeta}$ in the complex $\zeta$-plane under the transformation $\zeta\mapsto w=e^{-\epsilon\zeta}$.
The contour $\Sigma_w'$ can be chosen in such a way that $C_{\zeta}$ will be a clockwise oriented closed contour encircling all the points
$-\left(1+\nu_1\right)$, $\ldots$, $-\left(m_1-n\right)$; $\ldots$; $-\left(1+\nu_s\right)$, $\ldots$, $-\left(m_s-n\right)$ of the negative
real axis, and leaving $\zeta=0$ on its right.

Now let us consider the integral over $z$ in formula (\ref{K1}). We note that $u\geq -n$, so the residue of the integrand at infinity is zero.
Moreover, since $\beta>\alpha$,  and since $\Sigma_w'$ leaves $e^{\epsilon\beta}$ on its left, the singular point $z=\frac{1}{w}$ is situated inside $\Sigma_z$, for every $w\in\Sigma_w'$.
Thus we can deform $\Sigma_z$ through the extended complex plane into a new contour $\Sigma_z'$ encircling  the points $1$, $e^{\epsilon}$, $\ldots$, $e^{\epsilon(n-1)}$
of the complex $z$-plane once in the clockwise direction, and this deformation will not affect the integral.  The contour $\Sigma_z'$ can be obtained from a clockwise oriented contour $C_t$ by the transformation $t\mapsto z=e^{\epsilon t}$. Clearly, $\Sigma_z'$ can be chosen in such a way that $C_t$ will encircle the interval $[0,n-1]$, and will not intersect $C_{\zeta}$.

We make the change of
integration variables,
$$
z=e^{t\epsilon},\;\; w=e^{-\zeta\epsilon},
$$
and set
$$
u=-\frac{1}{\epsilon}\log x,\;\; v=-\frac{1}{\epsilon}\log y.
$$
This gives
\begin{equation}
\begin{split}
&K_{\Schur}^{\epsilon}\left(r,-\frac{1}{\epsilon}\log x;s,-\frac{1}{\epsilon}\log y\right)=\frac{\epsilon^2}{\left(2\pi i\right)^2}\oint\limits_{C_t}dt\oint\limits_{C_{\zeta}}d\zeta
\frac{x^t}{y^{\zeta}}\frac{1}{1-e^{-\epsilon(\zeta-t)}}\prod\limits_{a_0=1}^n\frac{1-e^{\epsilon\left(\zeta-a_0+1\right)}}{1-e^{\epsilon\left(t-a_0+1\right)}}\\
&\times\frac{\prod_{a_1=1+\nu_1}^{m_1-n}\left(1-e^{-\epsilon\left(t+a_1\right)}\right)\ldots \prod_{a_r=1+\nu_r}^{m_r-n}\left(1-e^{-\epsilon\left(t+a_r\right)}\right)}{\prod_{b_1=1+\nu_1}^{m_1-n}\left(1-e^{-\epsilon\left(\zeta+b_1\right)}\right)
\ldots\prod_{b_s=1+\nu_s}^{m_s-n}\left(1-e^{-\epsilon\left(\zeta+b_s\right)}\right)},
\end{split}
\end{equation}
where the integration contours are chosen as in the statement of Proposition \ref{THEOREMCorrelationKernel}.

Set
\begin{equation}\label{functiong}
\begin{split}
g\left(t,\zeta;\epsilon\right)=\frac{\epsilon}{1-e^{-\epsilon(\zeta-t)}}\prod\limits_{a_0=1}^n
\frac{\frac{1-e^{\epsilon\left(\zeta-a_0+1\right)}}{\epsilon}}{\frac{1-e^{\epsilon\left(t-a_0+1\right)}}{\epsilon}}
\frac{\prod_{a_1=1+\nu_1}^{m_1-n}\left(\frac{1-e^{-\epsilon\left(t+a_1\right)}}{\epsilon}\right)
\ldots\prod_{a_r=1+\nu_1}^{m_r-n}\left(\frac{1-e^{-\epsilon\left(t+a_r\right)}}{\epsilon}\right)}{
\prod_{b_1=1+\nu_1}^{m_1-n}\left(\frac{1-e^{-\epsilon\left(\zeta+b_1\right)}}{\epsilon}\right)\ldots
\prod_{b_s=1+\nu_s}^{m_s-n}\left(\frac{1-e^{-\epsilon\left(\zeta+b_s\right)}}{\epsilon}\right)}
\end{split}
\end{equation}
As $\epsilon\rightarrow 0+$, the function $g\left(t,\zeta;\epsilon\right)$ converges uniformly
(with respect to $\zeta$ on $C_{\zeta}$, and with respect to $t$ on $C_t$) to
\begin{equation}\label{ex}
\frac{1}{\zeta-t}\prod\limits_{a_0=1}^n\frac{\zeta-a_0+1}{t-a_0+1}
\frac{\prod_{a_1=1+\nu_1}^{m_1-n}\left(t+a_1\right)\ldots\prod_{a_r=1+\nu_r}^{m_r-n}\left(t+a_r\right)}{
\prod_{b_1=1+\nu_1}^{m_1-n}\left(\zeta+b_1\right)\ldots\prod_{b_s=1+\nu_s}^{m_s-n}\left(\zeta+b_s\right)},
\end{equation}
as $\epsilon\rightarrow 0+$.

Now we define
$$
\hat{K}_{\Schur}^{\epsilon}\left(r,u;s,v\right)=
\frac{\prod\limits_{k=1}^s\epsilon^{m_k-n-\nu_k}}{\prod\limits_{k=1}^r\epsilon^{m_k-n-\nu_k}}K_{\Schur}^{\epsilon}\left(r,u;s,v\right).
$$
Clearly, the kernels $\hat{K}_{\Schur}^{\epsilon}\left(r,u;s,v\right)$ and $K_{\Schur}^{\epsilon}\left(r,u;s,v\right)$ are equivalent.
Let us consider the limit
\begin{equation}\label{limit}
\underset{\epsilon\rightarrow 0+}{\lim}\left[\frac{1}{\epsilon y}\hat{K}_{\Schur}^{\epsilon}\left(r,-\frac{1}{\epsilon}\log x;s,-\frac{1}{\epsilon}\log y\right)\right].
\end{equation}
The fact that $g(t,\zeta;\epsilon)$ converges uniformly (with respect to $\zeta$ on $C_{\zeta}$, and with respect to $t$ on $C_t$) to expression (\ref{ex})
enables us to interchange the limit and the integrals, and to compute (\ref{limit}) explicitly.
By formula (\ref{Conv1}) this limit is equal to $K_{n,p,l=1}(r,x;s,y)$.  Thus we have
\begin{equation}\label{Conv3}
\begin{split}
&K_{n,p,l=1}(r,x;s,y)=
\frac{1}{\left(2\pi i\right)^2}
\oint\limits_{C_t}dt\oint\limits_{C_{\zeta}}d\zeta\frac{x^t}{y^{\zeta+1}(\zeta-t)}\prod_{a_0=1}^n\frac{\zeta-a_0+1}{t-a_0+1}\\
&\times\frac{\prod_{a_1=1+\nu_1}^{m_1-n}\left(a_1+t\right)\prod_{a_2=1+\nu_2}^{m_2-n}\left(a_2+t\right)\ldots
\prod_{a_r=1+\nu_r}^{m_r-n}\left(a_r+t\right)}{\prod_{b_1=1+\nu_1}^{m_1-n}\left(b_1+\zeta\right)\prod_{b_2=1+\nu_2}^{m_2-n}\left(b_2+\zeta\right)\ldots
\prod_{b_s=1+\nu_s}^{m_s-n}\left(b_s+\zeta\right)},
\end{split}
\end{equation}
where $r\geq s$. We rewrite the products inside the integrals above in terms of Gamma functions as follows
\begin{equation}
\prod_{a_0=1}^n\frac{\zeta-a_0+1}{t-a_0+1}=\frac{\Gamma(\zeta+1)}{\Gamma(\zeta+1-n)}\frac{\Gamma(t+1-n)}{\Gamma(t+1)},
\end{equation}
\begin{equation}
\begin{split}
&\prod_{a_1=1+\nu_1}^{m_1-n}\left(a_1+t\right)=\frac{\Gamma\left(t+m_1-n+1\right)}{\Gamma\left(t+\nu_1+1\right)},
\ldots,
\prod_{a_r=1+\nu_r}^{m_r-n}\left(a_r+t\right)=\frac{\Gamma\left(t+m_r-n+1\right)}{\Gamma\left(t+\nu_r+1\right)},\\
\end{split}
\end{equation}
and
\begin{equation}
\begin{split}
&\prod_{b_1=1+\nu_1}^{m_1-n}\left(b_1+\zeta\right)=\frac{\Gamma\left(\zeta+m_1-n+1\right)}{\Gamma\left(\zeta+\nu_1+1\right)},
\ldots,\prod_{b_s=1+\nu_s}^{m_s-n}\left(b_s+\zeta\right)=\frac{\Gamma\left(\zeta+m_s-n+1\right)}{\Gamma\left(\zeta+\nu_s+1\right)}.
\end{split}
\end{equation}
Taking this into account we see that the righthand side of equation (\ref{Conv3}) can be rewritten as  that of equation (\ref{CorrelationKernelProductProcess}).
This proves Proposition \ref{THEOREMCorrelationKernel} for $r\geq s$ (and $l=1$).

Assume that $r<s$. In this case we can choose both $\Sigma_z$ and  $\Sigma_w$ as counterclockwise circle contours
whose centers are at $0$, and whose radii are less than 1, see Theorem 2.2 in Ref. \cite{BorodinRains}.
Let us agree that $|z|<|w|<1$ for all $z\in\Sigma_z$, and $w\in\Sigma_{w}$. In addition, we will choose $\Sigma_z$ in such a way that all the points
$e^{-\left(1+\nu_1\right)\epsilon}$, $\ldots$, $e^{-\left(m_1-n\right)\epsilon}$; $\ldots$; $e^{-\left(1+\nu_s\right)\epsilon}$, $\ldots$, $e^{-\left(m_s-n\right)\epsilon}$
will be situated inside $\Sigma_z$.

We will deform the contour $\Sigma_w$ through the extended complex plane into a new contour $\Sigma_w'$
encircling  all the points $e^{\left(1+\nu_1\right)\epsilon}$, $\ldots$, $e^{\left(m_1-n\right)\epsilon}$; $\ldots$; $e^{\left(1+\nu_s\right)\epsilon}$, $\ldots$, $e^{\left(m_s-n\right)\epsilon}$ in the clockwise direction, and leaving the points $w=\frac{1}{z}$, $z\in\Sigma_z$ outside.
As we  deform $\Sigma_w$ into $\Sigma_w'$, we should pick up the contribution at the residue at $w=\frac{1}{z}$. Thus we rewrite  equation (\ref{K1}) as
\begin{equation}\label{K2}
\begin{split}
&K_{\Schur}^{\epsilon}\left(r,u;s,v\right)=
-\frac{1}{2\pi i}\oint\limits_{\Sigma_z}\frac{dz}{z^{u-v+1}}\frac{1}{\prod_{b_{r+1}=1+\nu_{r+1}}^{m_{r+1}-n}\left(1-e^{-b_{r+1}\epsilon}z^{-1}\right)
\ldots\prod_{b_s=1+\nu_s}^{m_s-n}\left(1-e^{-b_s\epsilon}z^{-1}\right)}\\
&+\frac{1}{\left(2\pi i\right)^2}
\oint\limits_{\Sigma_z}\oint\limits_{\Sigma_w'}\frac{dzdw}{(zw-1)z^{u+1}w^{v+1}}\prod_{a_0=1}^n\frac{1-e^{-\left(a_0-1\right)\epsilon}w^{-1}}{1-e^{-\left(a_0-1\right)\epsilon}z}\\
&\times\frac{\prod_{a_1=1+\nu_1}^{m_1-n}\left(1-e^{-a_1\epsilon}z^{-1}\right)
\ldots\prod_{a_r=1+\nu_r}^{m_r-n}\left(1-e^{-a_r\epsilon}z^{-1}\right)}{\prod_{b_1=1+\nu_1}^{m_1-n}\left(1-e^{-b_1\epsilon}w\right)
\ldots\prod_{b_s=1+\nu_s}^{m_s-n}\left(1-e^{-b_s\epsilon}w\right)}.
\end{split}
\end{equation}
Denote by $K_{\Schur}^{\epsilon, I}\left(r,u;s,v\right)$ the first term in the right-hand side of the above equation, and by
$K_{\Schur}^{\epsilon, II}\left(r,u;s,v\right)$ the second term, so that
\begin{equation}\label{KI}
K_{\Schur}^{\epsilon, I}\left(r,u;s,v\right)=
-\frac{1}{2\pi i}\oint\limits_{\Sigma_z}\frac{dz}{z^{u-v+1}}\frac{1}{\prod_{b_{r+1}=1+\nu_{r+1}}^{m_{r+1}-n}\left(1-e^{-b_{r+1}\epsilon}z^{-1}\right)
\ldots\prod_{b_s=1+\nu_s}^{m_s-n}\left(1-e^{-b_s\epsilon}z^{-1}\right)},
\end{equation}
and
\begin{equation}
\begin{split}
&K_{\Schur}^{\epsilon, II}\left(r,u;s,v\right)=
\frac{1}{\left(2\pi i\right)^2}
\oint\limits_{\Sigma_z}\oint\limits_{\Sigma_w'}\frac{dzdw}{(zw-1)z^{u+1}w^{v+1}}\prod_{a_0=1}^n\frac{1-e^{-\left(a_0-1\right)\epsilon}w^{-1}}{1-e^{-\left(a_0-1\right)\epsilon}z}\\
&\times\frac{\prod_{a_1=1+\nu_1}^{m_1-n}\left(1-e^{-a_1\epsilon}z^{-1}\right)
\ldots\prod_{a_r=1+\nu_r}^{m_r-n}\left(1-e^{-a_r\epsilon}z^{-1}\right)}{\prod_{b_1=1+\nu_1}^{m_1-n}\left(1-e^{-b_1\epsilon}w\right)
\ldots\prod_{b_s=1+\nu_s}^{m_s-n}\left(1-e^{-b_s\epsilon}w\right)}.
\end{split}
\end{equation}
In the formula for $K_{\Schur}^{\epsilon, II}\left(r,u;s,v\right)$
we deform the contour $\Sigma_z$ through the extended complex plane into a new contour $\Sigma_z'$ encircling the points
$1$, $e^{\epsilon}$, $\ldots$, $e^{(n-1)\epsilon}$ in the clockwise direction.
Since we agree that all the points $e^{-\left(1+\nu_1\right)\epsilon}$, $\ldots$, $e^{-\left(m_1-n\right)\epsilon}$
$\ldots$; $e^{-\left(1+\nu_s\right)\epsilon}$, $\ldots$, $e^{-\left(m_s-n\right)\epsilon}$
are situated inside $\Sigma_z$, the contour $\Sigma_w'$ can be chosen such that this deformation will not affect
the value of $K_{\Schur}^{\epsilon, II}\left(r,u;s,v\right)$. The contours $\Sigma_w'$ and $\Sigma_z'$ can be viewed as images of  contours $C_{\zeta}$, $C_{t}$
under the transformations $\zeta\mapsto w=e^{-\epsilon\zeta}$, and $t\rightarrow z=e^{\epsilon t}$, and the contours $C_{\zeta}$, $C_t$ will be those described in the statement of the Proposition.

Set
$$
\hat{K}_{\Schur}^{\epsilon,I}\left(r,u;s,v\right)=
\frac{\prod\limits_{k=1}^s\epsilon^{m_k-n-\nu_k}}{\prod\limits_{k=1}^r\epsilon^{m_k-n-\nu_k}}K_{\Schur}^{\epsilon,I}\left(r,u;s,v\right),
$$
$$
\hat{K}_{\Schur}^{\epsilon,II}\left(r,u;s,v\right)=
\frac{\prod\limits_{k=1}^s\epsilon^{m_k-n-\nu_k}}{\prod\limits_{k=1}^r\epsilon^{m_k-n-\nu_k}}K_{\Schur}^{\epsilon,II}\left(r,u;s,v\right),
$$
and
$$
\hat{K}_{\Schur}^{\epsilon}\left(r,u;s,v\right)=\hat{K}_{\Schur}^{\epsilon,I}\left(r,u;s,v\right)+\hat{K}_{\Schur}^{\epsilon,II}\left(r,u;s,v\right).
$$
After the change of variables we find
\begin{equation}\label{K3}
\begin{split}
&\widehat{K}_{\Schur}^{\epsilon,II}\left(r,-\frac{1}{\epsilon}\log x;s,-\frac{1}{\epsilon}\log y\right)=
\frac{\epsilon}{\left(2\pi i\right)^2}
\oint\limits_{C_t}dt\oint\limits_{C_{\zeta}}d\zeta\frac{x^t}{y^{\zeta}}g(t,\zeta;\epsilon),
\end{split}
\end{equation}
where $g(t,\zeta;\epsilon)$ is defined by equation (\ref{functiong}). We thus have
\begin{equation}\label{ZD}
\begin{split}
&\underset{\epsilon\rightarrow 0+}{\lim}\left[\frac{1}{\epsilon y}\widehat{K}_{\Schur}^{\epsilon,II}\left(r,-\frac{1}{\epsilon}\log x;s,-\frac{1}{\epsilon}\log y\right)\right]=
\frac{1}{\left(2\pi i\right)^2}
\oint\limits_{C_t}dt\oint\limits_{C_{\zeta}}d\zeta\frac{x^t}{y^{\zeta+1}(\zeta-t)}\prod_{a_0=1}^n\frac{\zeta-a_0+1}{t-a_0+1}\\
&\times\frac{\prod_{a_1=1+\nu_1}^{m_1-n}\left(a_1+t\right)\prod_{a_2=1+\nu_2}^{m_2-n}\left(a_2+t\right)\ldots
\prod_{a_r=1+\nu_r}^{m_r-n}\left(a_r+t\right)}{\prod_{b_1=1+\nu_1}^{m_1-n}\left(b_1+\zeta\right)\prod_{b_2=1+\nu_2}^{m_2-n}\left(b_2+\zeta\right)\ldots
\prod_{b_s=1+\nu_s}^{m_s-n}\left(b_s+\zeta\right)},
\end{split}
\end{equation}
where again we have used the uniform convergence of $g(t,\zeta;\epsilon)$ to take the limit inside the integrals.
Note that (as in the case $r\geq s$)  the right-hand side of equation (\ref{ZD}) can be written as the second term in the righthand side
of equation (\ref{CorrelationKernelProductProcess}).

Now consider the formula for $\hat{K}_{\Schur}^{\epsilon,I}\left(r,u;s,v\right)$ given by (\ref{KI}).
Assume that $u>v$. In this case the residue at infinity is equal to zero, and all finite poles are situated inside the contour $\Sigma_z$.
This implies that $\hat{K}_{\Schur}^{\epsilon,I}\left(r,u;s,v\right)$ is equal to zero for $u>v$.
If $u\leq v$, then the residue at $0$ is equal to $0$, and we can deform $\Sigma_z$ into a
contour $\Sigma_z'$ encircling the poles $e^{-\left(1+\nu_{r+1}\right)\epsilon}$, $\ldots$, $e^{-\left(m_{r+1}-n\right)\epsilon}$;
$\ldots;$ $e^{-\left(1+\nu_{s}\right)\epsilon}$, $\ldots$, $e^{-\left(m_{s}-n\right)\epsilon}$ once in the counterclockwise direction,
and leaving $0$ on the left.
The contour $\Sigma_z'$ can be viewed as an image of a contour $C_t$ under the transformation
$t\mapsto z=e^{\epsilon t}$, where $C_t$ will be that described in the statement of the Proposition.
After the change of variables we find that
\begin{equation}
\hat{K}_{\Schur}^{\epsilon,I}\left(r,u;s,v\right)=-\frac{\epsilon}{2\pi i}\oint\limits_{C_t}dt
\left(\frac{y}{x}\right)^{-t}
\frac{1}{\prod_{b_{r+1}=1+\nu_{r+1}}^{m_{r+1}-n}\frac{1-e^{\epsilon\left(t+b_{r+1}\right)}}{\epsilon}\ldots
\prod_{b_{s}=1+\nu_{s}}^{m_{s}-n}\frac{1-e^{\epsilon\left(t+b_s\right)}}{\epsilon}},
\end{equation}
where $y<x$. We have
\begin{equation}\label{Z}
\begin{split}
&\underset{\epsilon\rightarrow 0+}{\lim}\left[\frac{1}{\epsilon y}\widehat{K}_{\Schur}^{\epsilon,I}\left(r,-\frac{1}{\epsilon}\log x;s,-\frac{1}{\epsilon}\log y\right)\right]\\
&=-\frac{1}{2\pi i}\oint\limits_{C_{t}}dt\frac{x^{t}}{y^{t+1}}
\frac{1}{\prod\limits_{b_{r+1}=1+\nu_{r+1}}^{m_{r+1}-n}\left(b_{r+1}+t\right)\ldots\prod\limits_{b_{s}=1+\nu_{s}}^{m_{s}-n}\left(b_s+t\right)}\\
&=-\frac{1}{2\pi i}\oint\limits_{C_{t}}dt\frac{x^{t}}{y^{t+1}}
\frac{\Gamma\left(t+\nu_{r+1}+1\right)\ldots \Gamma\left(t+\nu_s+1\right)}{\Gamma\left(t+m_{r+1}-n+1\right)\ldots\Gamma\left(t+m_s-n+1\right)}\\
&=-\frac{1}{x}G^{s-r,0}_{s-r,s-r}\left(\begin{array}{ccc}
                                         m_{r+1}-n, & \ldots, & m_s-n \\
                                         \nu_{r+1}, & \ldots, & \nu_s
                                       \end{array}
\biggr|\frac{y}{x}\right),
\end{split}
\end{equation}
where $y<x$. As  follows from equations (2.22), (2.23) and (2.24) in Kieburg, Kuijlaars, and Stivigny \cite{KieburgKuijlaarsStivigny}, the function
$$
G^{s-r,0}_{s-r,s-r}\left(\begin{array}{ccc}
                                         m_{r+1}-n, & \ldots, & m_s-n \\
                                         \nu_{r+1}, & \ldots, & \nu_s
                                       \end{array}
\biggr|\frac{y}{x}\right)
$$
is equal to $0$ for $y\geq x$. Therefore, equation
\begin{equation}
\begin{split}
&\underset{\epsilon\rightarrow 0+}{\lim}\left[\frac{1}{\epsilon y}\widehat{K}_{\Schur}^{\epsilon,I}\left(r,-\frac{1}{\epsilon}\log x;s,-\frac{1}{\epsilon}\log y\right)\right]\\
&=-\frac{1}{x}G^{s-r,0}_{s-r,s-r}\left(\begin{array}{ccc}
                                         m_{r+1}-n, & \ldots, & m_s-n \\
                                         \nu_{r+1}, & \ldots, & \nu_s
                                       \end{array}
\biggr|\frac{y}{x}\right)
\end{split}
\end{equation}
holds true for $y\geq x$ as well.
Finally, formula (\ref{Conv1}), equations (\ref{ZD}) and (\ref{Z}) give the desired formula for the correlation kernel in the case $r<s$.
\qed
\section{Proof of Proposition \ref{PropositionEffectiveSchurMeasure} and Theorem \ref{MainTheorem}}\label{Section7}
\textit{Proof of Proposition} \ref{PropositionEffectiveSchurMeasure}\\
In order to compute the probability of the point configuration (\ref{PoinConfiguration}), we need to compute the projection
of the Schur process associated with  a skew plane partition $\Pi$ to the diagrams $\lambda^{(\alpha_1)}$, $\ldots$,$\lambda^{(\alpha_p)}$.
It is convenient to obtain first the projection on $\lambda^{(\beta_{2l+1})}$, $\ldots$, $\lambda^{(\beta_2)}$, $\lambda^{(\alpha_1)}$, $\ldots$,$\lambda^{(\alpha_p)}$,
see Figure \ref{Figure_corners}.
Using equation (\ref{SummationFormula2}) we find
\begin{equation}
\begin{split}
&\frac{1}{Z'}
s_{\lambda^{(\beta_{2l+1})}}\left(q^{-1},\ldots,q^{-(\beta_{2l+1}-1)}\right)
s_{\lambda^{(\beta_{2l+1})}/\lambda^{(\beta_{2l})}}\left(q^{\beta_{2l+1}},\ldots,q^{\beta_{2l}-1}\right)\\
& \times s_{\lambda^{(\beta_{2l-1})}/\lambda^{(\beta_{2l})}}\left(q^{-\beta_{2l}},\ldots,q^{-(\beta_{2l-1}-1)}\right)
s_{\lambda^{(\beta_{2l-1})}/\lambda^{(\beta_{2l-2})}}\left(q^{\beta_{2l-1}},\ldots,q^{\beta_{2l-2}-1}\right)\\
&\vdots\\
&\times s_{\lambda^{(\beta_{3})}/\lambda^{(\beta_{4})}}\left(q^{-\beta_{4}},\ldots,q^{-(\beta_{3}-1)}\right)
s_{\lambda^{(\beta_{3})}/\lambda^{(\beta_{2})}}\left(q^{\beta_{3}},\ldots,q^{\beta_{2}-1}\right)\\
&\times s_{\lambda^{(\alpha_1)}/\lambda^{(\beta_{2})}}\left(q^{-\beta_{2}},\ldots,q^{-(\alpha_{1}-1)}\right)
s_{\lambda^{(\alpha_{2})}/\lambda^{(\alpha_{1})}}\left(q^{-\alpha_1},\ldots,q^{-(\alpha_{2}-1)}\right)\\
&\vdots\\
&\times s_{\lambda^{(\alpha_{p})}/\lambda^{(\alpha_{p-1})}}\left(q^{-\alpha_{p-1}},\ldots,q^{-(\alpha_{p}-1)}\right)
s_{\lambda^{(\beta_{1})}/\lambda^{(\alpha_{p})}}\left(q^{-\alpha_p},\ldots,q^{-(\beta_{1}-1)}\right)\\
&\times s_{\lambda^{(\beta_{1})}}\left(q^{\beta_1},\ldots,q^{A+B}\right).
\end{split}
\end{equation}
Equations (\ref{SummationFormula1}) and (\ref{SummationFormula2}) enable us to sum over
the Young diagrams $\lambda^{(\beta_{2l+1})}$,  $\lambda^{(\beta_{2l})}$, $\ldots$, $\lambda^{(\beta_{2})}$, and
$\lambda^{(\beta_{1})}$. The result is
\begin{equation}
\begin{split}
&\frac{1}{Z''}s_{\lambda^{(\alpha_1)}}\left(q^{-1},\ldots,q^{-(\beta_{2l+1}-1)};q^{-\beta_{2l}},\ldots,q^{-(\beta_{2l-1}-1)};
\ldots;q^{-\beta_2},\ldots,q^{-(\alpha_1-1)}\right)\\
&\times s_{\lambda^{(\alpha_{2})}/\lambda^{(\alpha_{1})}}\left(q^{-\alpha_1},\ldots,q^{-(\alpha_{2}-1)}\right)
\ldots s_{\lambda^{(\alpha_{p})}/\lambda^{(\alpha_{p-1})}}\left(q^{-\alpha_{p-1}},\ldots,q^{-(\alpha_{p}-1)}\right)\\
&\times s_{\lambda^{(\alpha_{p})}}\left(q^{\beta_1},\ldots,q^{A+B}\right).
\end{split}
\end{equation}
Taking into account the homogeneity of the Schur polynomials, and noting that $\beta_1=A+\pi_1+1$, we see that the expression  above can be rewritten
as in the statement of the Proposition \ref{PropositionEffectiveSchurMeasure}.
\qed\\
\textit{Proof of Theorem} \ref{PropositionEffectiveSchurMeasure}.\\
Proposition \ref{PropositionEffectiveSchurMeasure} says that the probability of the point configuration
$$
\left\{\left(1,\lambda_i^{(\alpha_1)}-i\right)\right\}_{i\geq 1}\cup\ldots\cup\left\{\left(1,\lambda_i^{(\alpha_p)}-i\right)\right\}_{i\geq 1}
$$
can be written as a product of skew Schur functions, see equation (\ref{EffectiveSchurMeasure1}). If $q=e^{-\epsilon}$, the parameters $n$; $m_1,\ldots,m_{l+p-1}$;
$\nu_1$, $\ldots$, $\nu_{l+p-1}$ are related with parameters $A$, $B$; $\pi_1$; $\alpha_1$, $\ldots$, $\alpha_p$; $\beta_1$, $\ldots$, $\beta_{2l+1}$ as
in the statement of Theorem \ref{MainTheorem}, and $\lambda^{(\alpha_1)}$, $\ldots$, $\lambda^{(\alpha_p)}$ are identified with
$\lambda^{(1)},\ldots,\lambda^{(p)}$, then the probability measure defined by equation (\ref{EffectiveSchurMeasure}) turns into the probability measure defined by equation (\ref{EffectiveSchurMeasure1}). Application of Proposition \ref{PropositionGeneralLimit} gives the result.
\qed

\section{Alternative proof through symmetric functions and zonal polynomials}

\label{Section_spherical}

In this section we sketch an alternative path to derive Proposition \ref{PropositionGeneralLimit}.
For the clarity of the exposition we only detail one simplest case. At the end of the section we outline ways for generalizations.

\noindent {\bf Simplest case: product of two $2\times 2$ matrices.} Let $U_1$ and $U_2$ be two $4\times 4$ independent Haar-distributed random unitary complex
matrices. Let $T_1$ and $T_2$ be principal $2\times 2$ corners of $U_1$ and $U_2$, respectively.
Our aim is to link the distribution of the squared singular values of $T_1 T_2$ to a Schur measure.

First, note that $T_1$ and $T_2$ are almost surely non-degenerate. The squared singular values of
$T_1 T_2$ are eigenvalues of $T_1 T_2 T_2^* T_1^*$. Since eigenvalues are preserved under
conjugations, they are the same as the eigenvalues of $(T_1^* T_1)(T_2 T_2^*)$. Since, $T_1$ and
$T_1^*$ have the same distribution, we can further rewrite the law of interest as the distribution
of eigenvalues for the matrix $(T_1 T_1^*) (T_2 T_2^*)$.

Set $A=T_1 T_1^*$ and $B=T_2 T_2^*$. The (ordered) eigenvalues of $A$ are real numbers $a_1, a_2$
distributed with probability density
\begin{equation}
\label{eq_Jacobi_2}
 \rho(a_1,a_2)=12(a_2-a_1)^2,\quad 0\le a_1 \le a_2 \le 1.
\end{equation}
This computation is a particular case of the identification of singular values of a corner of С„ a
random unitary matrix with Jacobi ensemble, see \cite{Collins} and references therein. This is also the  $k=p=\ell=1$ case in Proposition \ref{PropositionProductTruncatedProcess}.  The
eigenvalues of $B$, $0\le b_1\le b_2\le 1$ have the same distribution.

Next, we fix $0<q<1$ and consider a distribution on pairs of integers $\lambda_1\ge \lambda_2 \ge
0$ with weight
\begin{equation}
\label{eq_Schur_2}
 P_q(\lambda_1,\lambda_2)=(1-q)(1-q^2)^2(1-q^3)\, s_{(\lambda_1,\lambda_2)}(1,q)\, s_{(\lambda_1,\lambda_2)}(q,q^2).
\end{equation}
On one hand, \eqref{eq_Schur_2} is a particular case of the Schur measure. On the other hand, the
explicit evaluations
$$
 s_{(\lambda_1,\lambda_2)}(1,q)=\frac{\det\begin{pmatrix} 1 & 1 \\ q^{\lambda_1+1} & q^{\lambda_2} \end{pmatrix}}{1-q}=
 \frac{q^{\lambda_2}-q^{\lambda_1+1}}{1-q},
 \qquad  s_{(\lambda_1,\lambda_2)}(q,q^2)= q^{\lambda_1+\lambda_2}\frac{q^{\lambda_2}-q^{\lambda_1+1}}{1-q},
$$
imply that upon the change of variables $q^{\lambda_i}=a_i$, $i=1,2$, in the limit $q\to 1$,
\eqref{eq_Schur_2} becomes \eqref{eq_Jacobi_2}. The same computation works for matrices of
arbitrary size and for any values of the random matrix parameter $\beta$, see
\cite[Section 3.1]{ForresterRains}, \cite[Theorem 2.8]{BG_GFF}.

The next step is to use the Cauchy--Littlewood identity (see \cite[Chapter I]{Macdonald})
\begin{equation}
\label{eq_Cauchy_2}
 \sum_{(\lambda_1,\lambda_2)}s_{(\lambda_1,\lambda_2)}(u_1,u_2) s_{(\lambda_1,\lambda_2,0^{k-2})}(v_1,v_2,\dots,v_k)=\prod_{i=1}^2\prod_{j=1}^k \frac{1}{1-u_i v_j},\quad k\ge 2.
\end{equation}
Equation \eqref{eq_Schur_2} and the identity \eqref{eq_Cauchy_2} lead to the expectation evaluation
\begin{equation}
\label{eq_expectation_Cauchy}
 \E_{P_q}\left[ \frac{s_{(\lambda_1,\lambda_2)}(u_1,u_2)} {s_{(\lambda_1,\lambda_2)}(1,q)}\right]=
  \prod_{i=1}^2 \frac{(1-q^{i})(1-q^{i+1})}{ (1-u_i q)(1-u_i q^2)}.
\end{equation}
Here $u_1$ and $u_2$  can be any complex numbers such that the series
defining the expectation absolutely converges. We make a particular choice,
$(u_1,u_2)=(q^{\mu_1+1},q^{\mu_2})$ for two integers $\mu_1\ge\mu_2\ge 0$. Then, using the
label--variable duality (which is an immediate consequence of the definition of the Schur polynomials as ratios of two determinants)
\begin{equation}
\label{eq_Schur_duality}
 \frac{s_{(\lambda_1,\lambda_2)}(q^{\mu_1+1},q^{\mu_2})} {s_{(\lambda_1,\lambda_2)}(1,q)}=
 \frac{s_{(\mu_1,\mu_2)}(q^{\lambda_1+1},q^{\lambda_2})} {s_{(\mu_1,\mu_2)}(1,q)},
\end{equation}
 we get
\begin{equation}
\label{eq_expectation_Cauchy_duality}
 \E_{P_q}\left[ \frac{s_{(\mu_1,\mu_2)}(q^{\lambda_1+1},q^{\lambda_2})} {s_{(\mu_1,\mu_2)}(1,q)} \right]=
 \prod_{i=1}^2 \frac{(1-q^{i})(1-q^{i+1})}{ (1-q^{\mu_i-i+3})(1-q^{\mu_i-i+4})},\qquad \mu_1\ge \mu_2\ge 0.
\end{equation}
Note that as one varies $\mu_1,\mu_2$, the left--hand side of \eqref{eq_expectation_Cauchy_duality}
uniquely determines all the moments of $P_q$--random vector $(q^{\lambda_1+1},q^{\lambda_2})$. Indeed, since $\lambda_1+1>\lambda_2$, it is enough to consider only symmetric linear combinations of moments, and those are finite combinations of Schur polynomials. Since $(q^{\lambda_1+1},q^{\lambda_2})$ is supported inside $[0,1]\times[0,1]$, the moments uniquely determine the distribution $P_q$. The conclusion is that
\eqref{eq_expectation_Cauchy_duality} is \emph{equivalent} to the definition \eqref{eq_Schur_2}.
Further, the same equivalency holds in the limit $q\to 1$, as $(q^{\lambda_1+1},q^{\lambda_2})\to
(a_1,a_2)$. Equation \eqref{eq_expectation_Cauchy_duality} becomes
\begin{equation}
\label{eq_expectation_RM}
 \E_{\rho(a_1,a_2)} \left[\frac{s_{(\mu_1,\mu_2)}(a_1,a_2)} {s_{(\mu_1,\mu_2)}(1,1)}\right]= \prod_{i=1}^2 \frac{i(i+1)}{ (\mu_i-i+3)(\mu_i-i+4)},
 \qquad \mu_1\ge \mu_2\ge 0.
\end{equation}

We can now describe what is happening with expectations \eqref{eq_expectation_Cauchy},
\eqref{eq_expectation_Cauchy_duality}, \eqref{eq_expectation_RM}, when we multiply the matrices.
The computation relies on the following integral identity:
\begin{equation}
\label{eq_characters_functional}
 \int_{U(2)} \frac{s_{(\mu_1,\mu_2)}( V U W U^{-1})}{s_{(\mu_1,\mu_2)}(1,1)}dU= \frac{s_{(\mu_1,\mu_2)}(v_1,v_2)}{s_{(\mu_1,\mu_2)}(1,1)} \cdot
\frac{s_{(\mu_1,\mu_2)}(w_1,w_2)}{s_{(\mu_1,\mu_2)}(1,1)}, \quad \mu_1\ge \mu_2\ge 0.
\end{equation}
In \eqref{eq_characters_functional}, the integration goes over the group $U(2)$ of $2\times 2$
unitary complex matrices, $V$ and $W$ are two fixed complex matrices with eigenvalues $(v_1,v_2)$
and $(w_1,w_2)$, respectively, and by $s_{(\mu_1,\mu_2)}( V U W U^{-1})$ we mean the Schur
polynomial $s_{(\mu_1,\mu_2)}$ evaluated on two eigenvalues of the matrix $V U W U^{-1}$. When $V$
and $W$ are unitary, the relation \eqref{eq_characters_functional} is known as \emph{the functional
equation} for the characters of $U(2)$.
More generally, \eqref{eq_characters_functional} is the identification of \emph{zonal polynomials}
of the symmetric space $GL(2;\mathbb C)/U(2)$ with Schur polynomials, see \cite[Chapter
VII]{Macdonald}, \cite[Section 13.4.3]{For}. For real and quaternion matrices, an analogue of
\eqref{eq_characters_functional}  holds with Schur polynomials replaced by the Jack polynomials.
Using $N=2$ in $U(N)$ also plays no special role in \eqref{eq_characters_functional} and the identity
extends to all $N>0$.

Coming back to $AB$, the matrices $A=T_1 T_1^*$ and $B=T_2 T_2^*$ are independent, and the distribution of each of them
is $U(2)$--invariant, with respect to the action by conjugation. In other words, while the
eigenvalues have a specific distribution \eqref{eq_Jacobi_2}, the eigenvectors are chosen uniformly
at random (in the set of all possible pairs of orthogonal unit vectors in $\mathbb C^2$). Thus, if
we plug $V=A$, $W=B$ into \eqref{eq_characters_functional} and take expectation with respect to $A$
and $B$, we get
$$
 \E  \left[\frac{s_{(\mu_1,\mu_2)}(AB)} {s_{(\mu_1,\mu_2)}(1,1)}\right]=
 \E  \left[\frac{s_{(\mu_1,\mu_2)}(A)} {s_{(\mu_1,\mu_2)}(1,1)}\right]
 \E  \left[\frac{s_{(\mu_1,\mu_2)}(B)} {s_{(\mu_1,\mu_2)}(1,1)}\right],  \quad \mu_1\ge \mu_2\ge 0.
$$
Combining with \eqref{eq_expectation_RM}, this implies
$$
\E \left[\frac{s_{(\mu_1,\mu_2)}(AB)} {s_{(\mu_1,\mu_2)}(1,1)}\right]=
\left(\prod_{i=1}^2 \frac{i(i+1)}{ (\mu_i-i+3)(\mu_i-i+4)}\right)^2,
 \qquad \mu_1\ge \mu_2\ge 0,
$$
which is (again by \eqref{eq_Cauchy_2} and \eqref{eq_Schur_duality})  precisely the $q\to 1$ limit
of
$$
\E_{\tilde P_q} \left[ \frac{s_{(\mu_1,\mu_2)}(q^{\nu_1+1},q^{\nu_2})}{s_{(\mu_1,\mu_2)}(1,q)}
\right], \quad \mu_1\ge \mu_2\ge 0,
$$ for the integral vector $\nu_1\ge\nu_2\ge 0$ distributed according to the Schur measure
\begin{equation}
\label{eq_Schur_for_product}
  \tilde P_q(\nu_1,\nu_2)=(1-q)^2(1-q^2)^4(1-q^3)^2\, s_{(\nu_1,\nu_2)}(1,q)\, s_{(\nu_1,\nu_2)}(q,q^2,q,q^2).
\end{equation}
We conclude that the eigenvalue distribution for $AB$ is the $q\to 1$ limit of
$(q^{\nu_1+1},q^{\nu_2})$  distributed according to the Schur measure \eqref{eq_Schur_for_product}.

\bigskip

\noindent {\bf Generalizations.} Let us discuss how to see the structure of the full Schur processes, rather than just its slices given by the Schur measures. Note that the transition between $P_q$ of \eqref{eq_Schur_2} and $\tilde P_q$ of \eqref{eq_Schur_for_product} can be seen as one step of  Markov chain on two-row Young diagrams $\lambda=(\lambda_1,\lambda_2)$ with transition probabilities $P(\lambda\to \nu)$ found from the decomposition
\begin{equation}
\label{eq_Markov_transition_Schur}
 \frac{s_\lambda(u_1,u_2)}{s_\lambda(1,q)} \cdot \left[  \prod_{i=1}^2 \frac{(1-q^{i})(1-q^{i+1})}{ (1-u_i q)(1-u_i q^2)}\right]=\sum_{\nu} P(\lambda\to\nu) \frac{s_\nu(u_1,u_2)}{s_\nu(1,q)}.
\end{equation}
Summing \eqref{eq_Markov_transition_Schur} using \eqref{eq_expectation_Cauchy}, we get $\sum_{\lambda} P_q(\lambda) P(\lambda\to\nu)=\tilde P_q(\nu)$. In the $q\to 1$ limit, the same structure of the Markov chain can be seen for the projection of the joint law of $A$ and $AB$ onto their eigenvalues; this is again a corollary of \eqref{eq_characters_functional}.

Comparison of \eqref{eq_Markov_transition_Schur} with the \emph{skew Cauchy identity} (see, e.g., \cite[Section I.5, Example 26]{Macdonald})  yields that $P(\lambda\to\nu)$ is given by the fomula
\begin{equation}
\label{eq_transition}
 P(\lambda\to\nu)=(1- q)(1-q^2)^2(1-q^3) \frac{s_\nu(1,q) s_{\nu/\lambda}(q,q^2) }{s_\lambda(1,q)}.
\end{equation}
Combining the definition of $P_q$ with \eqref{eq_transition}, we conclude that the two times Markov chain with initial state $P_q$, transitional probability $P(\lambda\to\nu)$ (and final state $\tilde P_q$) is the Schur process. Sending $q\to 1$, we see that that the joint law of squared singular values of $A$ and $AB$ is the continuous limit of this Schur process.

\medskip

At this moment we can generalize the argument to products of more matrices. Each additional factor gives another time step of the Markov chain. These transition probabilities generalize \eqref{eq_transition} and therefore, the link to Schur processes persists. Let us make a remark about the sizes of the matrices that are being multiplied. The computation leading to \eqref{eq_Jacobi_2} and its connection to \eqref{eq_Schur_2} can be generalized to rectangular corners of random unitary matrices of arbitrary sizes (and we can also deal with real/squaternion $\beta=1,4$ cases). The identity \eqref{eq_characters_functional} has similar extensions. However, when we start iterating \eqref{eq_characters_functional} it is convenient to assume that all the involved matrices have the same square size, as then the present arguments extend word-for-word. This is less general than the setting of Proposition \ref{PropositionGeneralLimit}. It is plausible that the arguments of this section can be adapted to the changing sizes as well,\footnote{When we pass between a rectangular matrix $T$ and its square counterpart $T T^*$ encapsulating the singular values, the following distinction becomes important: the matrices $T T^*$ and $T^* T$ have the same non-trivial eigenvalues, but additional $0$s show up because of different sizes. One needs to translate this into the language of Schur processes and expectations \eqref{eq_expectation_Cauchy}, \eqref{eq_expectation_Cauchy_duality}.} but we will not pursue this direction.

Finally, one can go beyond corners of unitary matrices and consider factors with more complicated unitarily--invariant distributions. We refer to \cite{GSun} for a progress in this direction.

\section{Appendix. Measures given by products of determinants, the Eynard-Mehta theorem, and the second proof of Proposition \ref{THEOREMCorrelationKernel}}\label{SECTIONEynardMehta}
The aim of this Appendix is to give another, more direct proof of Proposition \ref{THEOREMCorrelationKernel} based on an application of the Eynard-Mehta theorem \cite{EynardMehta}.
The starting point of this proof is the fact that the density of the product matrix process
under considerations is given by a product of determinants, see Proposition \ref{PropositionProductTruncatedProcess}.
This enables us to apply the Eynard-Mehta theorem to the product matrix process with truncated unitary matrices.
Although, the proof below is similar to the arguments of \cite{BorodinRains} and \cite{StrahovD}, we decided to reproduce it in the present setting for pedagogical reasons.

\subsection{The Eynard-Mehta theorem}
Let us first recall the formulation of the Eynard-Mehta theorem.
Let $n, p\geq 1$ be two fixed natural numbers, and let $\X_0$, $\X_{p+1}$ be two given sets. Let $\X$ be a complete separable
metric space, and consider a probability measure on $(\X^n)^p$ given by
\begin{equation}\label{ProductDeterminantsMeasure}
\begin{split}
\Prob_{n,p}(\underline{x})d\mu(\underline{x})&=\frac{1}{Z_{n,p}}\det\left(\phi_{0,1}(x_i^0,x_j^1)\right)_{i,j=1}^n
\det\left(\phi_{p,p+1}(x_i^p,x_j^{p+1})\right)_{i,j=1}^n\\
&
\times\prod\limits_{r=1}^{p-1}\det\left(\phi_{r,r+1}(x_i^r,x_j^{r+1})\right)_{i,j=1}^nd\mu(\underline{x}).
\end{split}
\end{equation}
In the formula  above, $Z_{n,p}$ is the normalization constant, the functions $\phi_{r,r+1}: \X\times \X\rightarrow \C$, $r=1,\ldots,p-1$
are given \textit{intermediate one-step  transition functions},  $\phi_{0,1}: \X_0\times \X\rightarrow \C$ is a given \textit{initial one-step transition function},
and $\phi_{p,p+1}: \X\times \X_{p+1}\rightarrow \C$ is a given \textit{final one-step transition function}. Also,
$$
\underline{x}=\left(x^1,\ldots,x^p\right)\in\left(\X^n\right)^p; \;\; x^r=\left(x^r_1,\ldots,x^r_n\right), r=1,\ldots, p,
$$
the vectors
$$
x^0=(x^0_1,\ldots,x^0_n)\in \X_0^n,\;\; x^{p+1}=(x^{p+1}_1,\ldots,x^{p+1}_n)\in \X_{p+1}^n,
$$
are fixed initial and final vectors,
and
$$
d\mu(\underline{x})=\prod\limits_{r=1}^p\prod\limits_{j=1}^nd\mu(x_j^r).
$$
Here $\mu$ is a given Borel measure on $\X$.
Given two transition functions $\phi$ and $\psi$ set
$$
\phi\ast\psi(x,y)=\int_{\X}\phi(x,t)\psi(t,y)d\mu(t).
$$

The theorem below is the Eynard-Mehta theorem.
\begin{thm}\label{TheoremEynardMehta} The probability measure $\Prob_{n,p}(\underline{x})d\mu(\underline{x})$
given by equation (\ref{ProductDeterminantsMeasure})  defines a determinantal point process on
$\left\{1,\ldots,p\right\}\times \X$. The correlation kernel of this determinantal point process, $K_{n,p}(r,x;s,y)$
(where $r,s\in\left\{1,\ldots,p\right\}$, and $x, y\in \X)$, is given by the formula
\begin{equation}\label{CorrelationKernelGeneralFormula}
K_{n,p}(r,x;s,y)=-\phi_{r,s}(x,y)+\sum\limits_{i,j=1}^n\phi_{r,p+1}(x,x_i^{p+1})\left(A^{-1}\right)_{i,j}\phi_{0,s}(x_j^0,y).
\end{equation}
The functions $\phi_{r,s}$, and the matrix $A=(a_{i,j})$ (where $i,j=1,\ldots,n$) are defined in terms of transition functions as follows \begin{equation}
\phi_{r,s}(x,y)=\left\{
                  \begin{array}{ll}
                    \left(\phi_{r,r+1}\ast\ldots\ast\phi_{s-1,s}\right)(x,y), & 0\leq r<s\leq p+1, \\
                    0, & r\geq s,
                  \end{array}
                \right.
\end{equation}
and
\begin{equation}
a_{i,j}=\phi_{0,p+1}(x_i^0,x_j^{p+1}).
\end{equation}
\end{thm}
\begin{remark}
 For the process defined by probability measure (\ref{ProductDeterminantsMeasure}), and  described by Theorem \ref{TheoremEynardMehta}, the correlation functions
can be written as determinants of block matrices, namely
\begin{equation}
\begin{split}
&\varrho_{k_1,\ldots,k_p}\left(x_1^1,\ldots,x_{k_1}^1;\ldots;x_1^p,\ldots,x_{k_p}^p\right)\\
&=\det\left[\begin{array}{ccc}
             \left(K_{n,p}(1,x_i^1;1,x_j^1)\right)_{i=1,\ldots,k_1}^{j=1,\ldots,k_1}
              & \ldots & \left(K_{n,p}(1,x_i^1;p,x_j^p)\right)_{i=1,\ldots,k_1}^{j=1,\ldots,k_p} \\
              \vdots &  &  \\
 \left(K_{n,p}(p,x_i^p;1,x_j^1)\right)_{i=1,\ldots,k_p}^{j=1,\ldots,k_1}
  & \ldots & \left(K_{n,p}(p,x_i^p;p,x_j^p)\right)_{i=1,\ldots,k_p}^{j=1,\ldots,k_p}
            \end{array}
\right],
\end{split}
\nonumber
\end{equation}
where $1\leq k_1,\ldots,k_p\leq n$, and for $1\leq l,r\leq p$
$$
\left(K_{n,p}(l,x_i^l;r,x_j^r)\right)_{i=1,\ldots,k_l}^{j=1,\ldots,k_r}
=\left(\begin{array}{ccc}
         K_{n,p}(l,x_1^l;r,x_1^r) & \ldots & K_{n,p}(l,x_1^l;r,x_{k_r}^r)\\
         \vdots &  &  \\
         K_{n,p}(l,x_{k_l}^l;r,x_1^r) & \ldots & K_{n,p}(l,x_{k_l}^l;r,x_{k_r}^r)
       \end{array}
\right).
$$
\end{remark}

In what follows  the functions
$$
\phi_{0,s}(i,y),\;\; 2\leq s\leq p,
$$
will be called \textit{initial transition functions}, and the functions
$$
\phi_{r,p+1}(x,j),\;\; 1\leq r\leq p-1,
$$
will be called  \textit{final transition functions}. In addition, the functions of the form
$$
\phi_{r,s}(x,y),\;\; 1\leq r\leq p-2,\;\; r+2\leq s\leq p,
$$
will be called \textit{intermediate transition functions.} Finally, the function
$$
\phi_{0,p+1}(i,j)
$$
will be called \textit{the total transition function}.

In order to prove Proposition \ref{THEOREMCorrelationKernel} we rewrite the density of the product matrix process
with truncated unitary matrices obtained in Proposition \ref{PropositionProductTruncatedProcess} as in the statement
of the Eynard-Mehta theorem, and obtain explicit expressions for the transition functions.
\subsection{Explicit formulae for the transition functions}\label{SectionFormulaeTransitionFunctions}
In our situation $\X_0=\left\{1,\ldots,n\right\}$, $\X_{p+1}=\left\{1,\ldots,n\right\}$, $\X=\R_{>0}$.
The initial one-step transition function is defined by
$$
\phi_{0,1}:\;\left\{1,\ldots,n\right\}\times\R_{>0}\longrightarrow\R;\; \phi_{0,1}(i,x)=w_i^{(l)}(x).
$$
The final one-step transition function is defined by
$$
\phi_{p,p+1}:\;\R_{>0}\times\left\{1,\ldots,n\right\}\longrightarrow\R;\; \phi_{p,p+1}(x,k)=x^{k-1}.
$$
In addition, the intermediate transition functions can be written as
$$
\phi_{r,r+1}:\;\R_{>0}\times\R_{>0}\longrightarrow\R,\;\; r=1,\ldots,p-1,
$$
$$
\phi_{r,r+1}(x,y)=y^{\nu_{r+l}}\left(x-y\right)_+^{m_{r+l}-n-\nu_{r+l}-1}x^{-m_{r+l}+n},
$$
where $(x-y)_+=\max\left(0,x-y\right)$.

For the initial transition functions $\phi_{0,s}, s=2,\ldots,p,$ we obtain the following recurrence
relation
\begin{equation}
\phi_{0,s}\left(i,x\right)=\int\limits_0^1\tau^{\nu_{s-1+l}}(1-\tau)^{m_{s-1+l}-n-\nu_{s-1+l}-1}\phi_{0,s-1}\left(i,\frac{x}{\tau}\right)\frac{d\tau}{\tau}.
\end{equation}
This is the recurrence relation for  $\phi_{0,s}(i,x)=w_i^{(s-1+l)}(x)$ (where $1\leq s\leq p$), see
Kieburg, Kuijlaars, and Stivigny, equation (2.23).

The total transition function $\phi_{0,p+1}(i,j)$ can be written as
\begin{equation}
\begin{split}
&\phi_{0,p+1}\left(i,j\right)=\int\limits_0^{\infty}\phi_{0,p}\left(i,p\right)\phi_{p,p+1}(t,j)dt=\int\limits_0^{\infty}t^{j-1}w_i^{(p-1+l)}(t)dt\\
&=c_{p-1+l}
\int\limits_0^{\infty}t^{j-1}
G_{p+l-1,p+l-1}^{p+l-1,0}\left(\begin{array}{cccc}
                           m_{p+l-1}-n, & \ldots, & m_2-n, & m_1-2n+i \\
                           \nu_{p+l-1}, & \ldots, & \nu_2, & \nu_1+i-1
                         \end{array}
\biggl\vert t\right)dt,
\end{split}
\end{equation}
where $c_{p-1+l}$ is defined by equation (\ref{cl}).
The Mellin transform of a Meijer $G$-function is
\begin{equation}
\int\limits_0^{\infty}dx x^{s-1}G_{p,q}^{m,n}\left(\begin{array}{cccc}
                           a_1, & \ldots, & a_p \\
                           b_1, & \ldots, & b_q
                         \end{array}
\biggl\vert xy\right)=\frac{1}{y^s}\frac{\prod\limits_{i=1}^m\Gamma\left(b_i+s\right)
\prod\limits_{j=1}^n\Gamma\left(1-a_j-s\right)}{\prod\limits_{k=m+1}^q\Gamma\left(1-b_k-s\right)
\prod\limits_{l=n+1}^p\Gamma\left(a_l+s\right)}.
\end{equation}
This gives an expression of the \textit{total transition function} $\phi_{0,p+1}(i,j)$
in terms of Gamma functions
\begin{equation}\label{TotalTransitionFunction}
\begin{split}
\phi_{0,p+1}\left(i,j\right)&=c_{p-1+l}
\frac{\prod_{k=2}^{p+l-1}\Gamma\left(\nu_k+j\right)\Gamma\left(\nu_1+i+j-1\right)}{\prod_{k=2}^{p+l-1}\Gamma\left(m_k-n+j\right)\Gamma\left(m_1-2n+i+j\right)}.
\end{split}
\end{equation}
Similar calculations give us the formula for the \textit{final transition functions}.
In particular, we can write
\begin{equation}
\begin{split}
&\phi_{p-1,p+1}(x,j)=\int\limits_0^{\infty}\phi_{p-1,p}(x,y)\phi_{p,p+1}(y,j)dy\\
&=\int\limits_0^{\infty}y^{\nu_{p-1+l}}\left(x-y\right)_+^{m_{p-1+l}-n-\nu_{p-1+l}-1}x^{-m_{p-1+l}+n}y^{j-1}dy\\
&=x^{-m_{p-1+l}+n}\int\limits_0^xy^{\nu_{p-1+l}}\left(x-y\right)^{m_{p-1+l}-n-\nu_{p-1+l}-1}y^{j-1}dy.
\end{split}
\end{equation}
The change of the integration variable $y=\tau x$ gives
\begin{equation}
\begin{split}
&\phi_{p-1,p+1}(x,j)=x^{j-1}\int\limits_0^1\tau^{\nu_{p-1+l}+j-1}\left(1-\tau\right)^{m_{p-1+l}-n-\nu_{p-1+l}-1}d\tau\\
&=x^{j-1}\frac{\Gamma\left(\nu_{p-1+l}+j\right)\Gamma\left(m_{p-1+l}-n-\nu_{p-1+l}\right)}{\Gamma\left(m_{p-1+l}-n+j\right)}.
\end{split}
\end{equation}
Repeating this calculation, we arrive to  a general formula for the final transition function
\begin{equation}\label{FinalTransitionFunctions}
\phi_{r,p+1}\left(x,j\right)=\prod\limits_{a=r}^{p-1}\frac{\Gamma\left(\nu_{a+l}+j\right)
\Gamma\left(m_{a+l}-n-\nu_{a+l}\right)}{\Gamma\left(m_{a+l}-n+j\right)}
x^{j-1},\;\;\; r\in\left\{1,\ldots,p-1\right\}.
\end{equation}

Finally, let us find the \textit{intermediate transition functions} $\phi_{r,s}(x,y)$.
Since
$$
G_{1,1}^{1,0}\left(\begin{array}{c}
                     a \\
                     b
                   \end{array}
\biggl\vert x\right)=\frac{(1-x)^{a-b-1}x^b}{\Gamma(a-b)},\;\;\; 0<x<1,
$$
we can rewrite the transition function $\phi_{r,r+1}(x,y)$ as
\begin{equation}\label{PhiMeijerRepresentation}
\phi_{r,r+1}(x,y)=\frac{1}{x}\Gamma\left(m_{r+l}-n-\nu_{r+l}\right)G_{1,1}^{1,0}\left(\begin{array}{c}
                     m_{r+l}-n \\
                     \nu_{r+l}
                   \end{array}
\biggl\vert \frac{y}{x}\right).
\end{equation}
In addition, we have the following recurrence relation
\begin{equation}\label{Grecurrence}
\begin{split}
&\int\limits_0^1x^{\nu_r}(1-x)^{m_r-n-\nu_r-1}G_{r-1,r-1}^{r-1,0}\left(\begin{array}{cccc}
                                                                        m_{r-1}-n, & \ldots, & m_2-n, & m_1-n \\
                                                                        \nu_{r-1}, & \ldots, & \nu_2, & \nu_{1}+n-1
                                                                      \end{array}
\biggl\vert\frac{y}{x}\right)\frac{dx}{x}\\
&=\Gamma\left(m_r-n-\nu_r\right)G_{r,r}^{r,0}\left(\begin{array}{cccc}
                                                                        m_{r}-n, & \ldots, & m_2-n, & m_1-n \\
                                                                        \nu_{r}, & \ldots, & \nu_2, & \nu_{1}+n-1
                                                                      \end{array}
\biggl\vert y\right),
\end{split}
\end{equation}
see Beals and Szmigielski \cite{Beals}, equation (5). Starting from (\ref{PhiMeijerRepresentation}), and applying
(\ref{Grecurrence}), we obtain
\begin{equation}\label{IntermediateTransitionFunctions}
\phi_{r,s}\left(x,y\right)=\frac{1}{x}\prod\limits_{k=r+1}^s\Gamma\left(m_{k+l-1}-n-\nu_{k+l-1}\right)
G_{s-r,s-r}^{s-r,0}\left(\begin{array}{cccc}
                          m_{r+l}-n, & \ldots, & m_{s+l-1}-n \\
                           \nu_{r+l}, & \ldots, & \nu_{s+l-1}
                         \end{array}
\biggl\vert \frac{y}{x}\right),
\end{equation}
where $1\leq r\leq p-2$, and  $r+2\leq s\leq p$.
\subsection{The inverse of $A$}
Here we find an explicit formula for $\left(A^{-1}\right)_{i,j}$ in the formula for the correlation kernel in
Theorem \ref{THEOREMCorrelationKernel}.  If $A=\left(a_{i,j}\right)_{i,j=1}^n$, then
$a_{i,j}$ is equal to the total transition function $\phi_{0,p+1}\left(i,j\right)$ given by equation (\ref{TotalTransitionFunction}).
We need to find the inverse of the matrix $\check{A}=\left(\check{a}_{i,j}\right)_{i,j=1}^n$ defined by
$$
\check{a}_{i,j}=\frac{\Gamma(\nu_1-1+i+j)}{\Gamma(m_1-2n+i+j)},\;\;\;\; i,j\in\left\{1,\ldots,n\right\}.
$$
\begin{prop}\label{PropositionZhangChenTheorem10} For $N=1,2,\ldots$ and $-\alpha,-\beta\in\C\setminus\mathbb{N}$, the matrix
$$
\left[\frac{(\alpha+1)_{i+j}}{(\alpha+\beta+2)_{i+j}}\right]_{i,j=0}^{N-1}
$$
is invertible  and its inverse matrix $\left(\gamma_{i,j}\right)_{i,j=0}^{N-1}$ is given by
\begin{equation}
\begin{split}
&\gamma_{i,j}=
\frac{(-1)^{i+j}\left(\alpha+\beta+1\right)_i\left(\alpha+\beta+1\right)_j}{\left(\alpha+1\right)_i\left(\alpha+1\right)_j\left(\alpha+\beta+1\right)}\\
&\times\sum\limits_{p=0}^{N-1}\frac{(2p+\alpha+\beta+1)(\alpha+1)_pp!}{(\alpha+\beta+1)_p(\beta+1)_p(p-i)!i!(p-j)!j!}
\left(\alpha+\beta+i+1\right)_p\left(\alpha+\beta+j+1\right)_p.
\end{split}
\end{equation}
\end{prop}
\begin{proof}
See Theorem 10 in  Zhang and Chen \cite{ZhangChen}.
\end{proof}
We apply Proposition \ref{PropositionZhangChenTheorem10}, and find
\begin{equation}
A^{-1}=\left(b_{i,j}\right)_{i,j=1}^n,
\end{equation}
where
\begin{equation}
\begin{split}
&b_{i,j}=\left(c_{p-1+l}\right)^{-1}\frac{\prod_{k=2}^{l+p-1}\Gamma\left(m_k-n+i\right)}{\prod_{k=2}^{l+p-1}\Gamma\left(\nu_k+i\right)}
\frac{\Gamma\left(m_1-2n+2\right)}{\Gamma\left(\nu_1+1\right)}\\
&\times(-1)^{i+j}\frac{\left(m_1-2n+1\right)_{i-1}\left(m_1-2n+1\right)_{j-1}}{\left(\nu_1+1\right)_{i-1}\left(\nu_1+1\right)_{j-1}\left(m_1-2n+1\right)}\\
&\times\sum\limits_{k=0}^{n-1}\frac{\left(2k+m_1-2n+1\right)\left(\nu_1+1\right)_k}{\left(m_1-2n+1\right)_k\left(m_1-2n-\nu_1+1\right)_k}
\frac{k!\left(m_1-2n+i\right)_k\left(m_1-2n+j\right)_k}{(k-i+1)!(i-1)!(k-j+1)!(j-1)!}.
\end{split}
\end{equation}
Now we can write the second term in the formula for the correlation kernel (equation (\ref{CorrelationKernelGeneralFormula})) as
\begin{equation}\label{KTILDA}
\begin{split}
&\widetilde{K}_{n,p}(r,x;s,y)=\sum\limits_{i,j=1}^n\phi_{r,p+1}(x,i)\left(A^{-1}\right)_{i,j}\phi_{0,s}(j,y)\\
&=\frac{\Gamma\left(m_1-2n+1\right)}{\Gamma\left(m_1-2n-\nu_1+1\right)\Gamma\left(\nu_1+1\right)}\frac{1}{\prod_{k=2}^{l+p-1}\Gamma\left(m_k-n-\nu_k\right)}\\
&\times\sum\limits_{k=0}^{n-1}\frac{k!\left(\nu_1+1\right)_k\left(2k+m_1-2n+1\right)}{\left(m_1-2n+1\right)_k\left(m_1-2n-\nu_1+1\right)_k}
P_{r,k}(x)Q_{s,k}(y),
\end{split}
\end{equation}
where
\begin{equation}
P_{r,k}(x)=\sum\limits_{i=0}^k\frac{(-1)^{k-i}\left(m_1-2n+i+1\right)_k\left(m_1-2n+1\right)_i\prod_{a=2}^{p+l-1}\Gamma\left(m_a-n+i+1\right)}{
(k-i)!i!\left(\nu_1+1\right)_i\prod_{a=2}^{p+l-1}\Gamma\left(\nu_a+i+1\right)}\phi_{r,p+1}(x,i+1),
\end{equation}
and
\begin{equation}\label{QasSum}
Q_{s,k}(y)=\sum\limits_{j=0}^k\frac{(-1)^k\left(m_1-2n+j+1\right)_k\left(m_1-2n+1\right)_j}{(k-j)!j!\left(\nu_1+1\right)_j}
\phi_{0,s}\left(j+1,y\right).
\end{equation}
\subsection{The contour integral representation for $P_{r,k}(x)$}
Using explicit formula for the final transition function $\phi_{r,p+1}(x,i)$ (see equation (\ref{FinalTransitionFunctions})) we rewrite
$P_{r,k}(x)$ as
\begin{equation}
\begin{split}
P_{r,k}(x)&=\frac{\Gamma\left(\nu_1+1\right)\prod_{a=r+l}^{p+l-1}\Gamma\left(m_a-n-\nu_a\right)}{\Gamma\left(m_1-2n+1\right)}\\
&\times\sum\limits_{i=0}^k\frac{(-1)^{k-i}}{(k-i)!i!}\frac{\Gamma\left(m_1-2n+i+k+1\right)
\prod_{a=2}^{r+l-1}\Gamma\left(m_a-n+i+1\right)}{\prod_{a=1}^{r+l-1}\Gamma\left(\nu_a+i+1\right)}x^i.
\end{split}
\end{equation}
The Residue Theorem gives the following contour integral representation for $P_{r,k}(x)$:
\begin{equation}\label{Prk}
\begin{split}
P_{r,k}(x)&=\frac{\Gamma\left(\nu_1+1\right)\prod_{a=r+l}^{p+l-1}\Gamma\left(m_{a}-n-\nu_{a}\right)}{\Gamma\left(m_1-2n+1\right)}\\
&\times\frac{1}{2\pi i}\oint\limits_{\Sigma_k}\frac{\Gamma(t-k)\Gamma\left(m_1-2n+t+k+1\right)
\prod_{a=2}^{r+l-1}\Gamma\left(m_a-n+t+1\right)}{\prod_{a=0}^{r+l-1}\Gamma\left(\nu_a+t+1\right)}x^tdt,
\end{split}
\end{equation}
where $\Sigma_k$ is a closed contour encircling the interval $[0,k]$ once in the positive direction, and where $\nu_0=0$.
\subsection{The contour integral representation for $Q_{s,k}(y)$}
Equation (\ref{QasSum}) together with the formula for the initial transition functions obtained in Section \ref{SectionFormulaeTransitionFunctions}
give
\begin{equation}\label{QasSum1}
\begin{split}
& Q_{s,k}(y)=\frac{\Gamma\left(m_1-2n-\nu_1+1\right)\Gamma\left(\nu_1+1\right)\prod_{a=2}^{s+l-1}\Gamma\left(m_a-n-\nu_a\right)}{\Gamma\left(m_1-2n+1\right)}
\\
&\times\sum\limits_{j=0}^k\frac{(-1)^{k-j}}{(k-j)!j!}
\frac{\Gamma\left(m_1-2n+j+k+1\right)}{\Gamma\left(\nu_1+1+j\right)}\\
&\times G_{s+l-1,s+l-1}^{s+l-1,0}\left(\begin{array}{cccc}
                     m_{s+l-1}-n, & \ldots, & m_2-n, & m_1-2n+j+1 \\
                     \nu_{s+l-1}, & \ldots, & \nu_2, & \nu_1+j
                   \end{array}
\biggl\vert y\right).
\end{split}
\end{equation}
The contour integral representation for the Meijer $G$-function inside the sum above is
\begin{equation}
\begin{split}
&G_{s+l-1,s+l-1}^{s+l-1,0}\left(\begin{array}{cccc}
                     m_{s+l-1}-n, & \ldots, & m_2-n, & m_1-2n+j+1 \\
                     \nu_{s+l-1}, & \ldots, & \nu_2, & \nu_1+j
                   \end{array}
\biggl\vert y\right)\\
&=\frac{1}{2\pi i}\int\limits_{C}\frac{\Gamma\left(\nu_1+j+u\right)\prod_{a=2}^{s+l-1}
\Gamma\left(\nu_a+u\right)}{\Gamma\left(m_1-2n+j+u+1\right)\prod_{a=2}^{s+l-1}\Gamma\left(m_a-n+u\right)}
y^{-u}du,
\end{split}
\end{equation}
where the contour $C$ is a positively oriented curve in the complex $u$-plane that starts and ends at $-\infty$, and encircles the negative real axis.
In equation (\ref{QasSum1}) the resulting sum is
\begin{equation}
\begin{split}
&\sum\limits_{j=0}^k\frac{(-1)^{k-j}}{(k-j)!j!}
\frac{\Gamma\left(m_1-2n+j+k+1\right)}{\Gamma\left(\nu_1+1+j\right)}
\frac{\Gamma\left(\nu_1+j+u\right)}{\Gamma\left(m_1-2n+j+u+1\right)}\\
&=\frac{(-1)^k}{k!}\frac{\Gamma\left(u+\nu_1\right)}{\Gamma\left(1+\nu_1\right)}
\frac{\Gamma\left(m_1-2n+k+1\right)}{\Gamma\left(m_1-2n+u+1\right)}
\sum\limits_{j=0}^k\frac{\left(-k\right)_j}{j!}\frac{\left(\nu_1+u\right)_j}{\left(\nu_1+1\right)_j}
\frac{\left(m_1-2n+k+1\right)_j}{\left(m_1-2n+k+1\right)_j}\\
&=\frac{(-1)^k}{k!}\frac{\Gamma\left(u+\nu_1\right)}{\Gamma\left(1+\nu_1\right)}
\frac{\Gamma\left(m_1-2n+k+1\right)}{\Gamma\left(m_1-2n+u+1\right)}\\
&\times
{}_3F_2\left(-k,u+\nu_1,m_1-2n+k+1;1+\nu_1,m_1-2n+u+1;1\right).
\end{split}
\nonumber
\end{equation}
The Pfaff-Saalsch$\ddot{\mbox{u}}$ltz Theorem says that
$$
{}_3F_2\left(-k,a,b;c,d;1\right)=\frac{(c-a)_k(c-b)_k}{(c)_k(c-a-b)_k},
$$
if the balanced condition, $c+d=1-k+a+b$, is satisfied, see, for example, Ismail \cite{Ismail}, Section 1.4. In our case
$$
a=u+\nu_1,\;b=l-2n+1+k,\; c=1+\nu_1,\; d=m_1-2n+1+u,
$$
and the balanced condition is satisfied. Thus we have
$$
{}_3F_2\left(-k,u+\nu_1,m_1-2n+1+k;1+\nu_1,m_1-2n+1+u;1\right)=\frac{(1-u)_k(\nu_1-m_1+2n-k)_k}{(1+\nu_1)_k(2n-u-m_1-k)_k}.
$$
Taking into account that
$$
\frac{(1-u)_k}{(2n-m_1-u-k)_k}=\frac{(u-k)_k}{(u-2n+m_1+1)_k},
$$
we obtain the formula
\begin{equation}\label{QspTruncation1}
\begin{split}
&Q_{s,k}(y)=(-1)^k\frac{\Gamma(m_1-2n-\nu_1+1)}{\Gamma(m_1-2n+1)}\frac{\Gamma(m_1-2n+1+k)(\nu_1-m_1+2n-k)_k}{k!(1+\nu_1)_k}\\
&\times\frac{\prod_{a=2}^{s+l-1}\Gamma\left(m_a-n-\nu_a\right)}{2\pi i}\int\limits_{C}\frac{\prod_{a=1}^{s+l-1}\Gamma\left(\nu_a+u\right)}{\prod_{a=2}^{s+l-1}\Gamma\left(m_a-n+u\right)}
\frac{(u-k)_k}{\left(u-2n+m_1+1\right)_k}\frac{y^{-u}du}{\Gamma\left(m_1-2n+1+u\right)}.
\end{split}
\end{equation}
\subsection{Derivation of the correlation kernel}
Equation (\ref{IntermediateTransitionFunctions}) gives the first term in equation (\ref{CorrelationKernelGeneralFormula}) for the correlation kernel.
 To write explicitly the second term in equation (\ref{CorrelationKernelGeneralFormula}) we insert  the formulae for $P_{r,k}$ (equation (\ref{Prk})),   and $Q_{s,k}$  (equation (\ref{QspTruncation1})) into equation (\ref{KTILDA}). After simplifications we see that the second term in equation (\ref{CorrelationKernelGeneralFormula})
 can be written as
\begin{equation}
\begin{split}
&\widetilde{K}_{n,p}\left(r,x;s,y\right)=\frac{\prod_{a=r+l}^{p+l-1}\Gamma\left(m_a-n-\nu_a\right)}{\prod_{a=s+l}^{p+l-1}\Gamma\left(m_a-n-\nu_a\right)}\\
&\times\frac{1}{(2\pi i)^2}
\oint\limits_{\Sigma_k}dt\int_{C}du
\frac{\prod_{a=2}^{r+l-1}\Gamma\left(m_a-n+t+1\right)}{\prod_{a=0}^{r+l-1}\Gamma\left(t+\nu_a+1\right)}
\frac{\prod_{a=0}^{s+l-1}\Gamma\left(\nu_a+u\right)}{\prod_{a=2}^{s+l-1}\Gamma\left(m_a-n+u\right)}
\\
&\times\sum\limits_{k=0}^{n-1}(m_1-2n+2k+1)\frac{\Gamma(t-k)\Gamma(t+m_1-2n+k+1)}{\Gamma(u-k)\Gamma(u+m_1-2n+k+1)}x^ty^{-u}.
\end{split}
\end{equation}
The sum inside the integral is the same  as that in Kuijlaars,  Stivigny \cite{KuijlaarsStivigny}
(see the proof of Proposition 4.4 in Kuijlaars,  Stivigny \cite{KuijlaarsStivigny}), and
the rest of the proof  is the same as that of
Proposition 4.4 and Theorem 4.7 in Kuijlaars,  Stivigny \cite{KuijlaarsStivigny}.
\qed



\begin{thebibliography}{99}
\bibitem{AdlerMoerbekeWang}Adler, M.; van Moerbeke, P.; Wang, D. Random matrix minor process
    related to percolation theory. Random Matrix Theory Appl. 2 (2013), no. 4, 1350008.


\bibitem{Akemann1} Akemann, G.; Burda, Z. Universal microscopic correlation functions for products
    of independent Ginibre matrices. J. Phys. A: Math. Theor. {\bf 45} (2012)  465201.

\bibitem{AkemannKieburgWei} Akemann, G.; Kieburg M.; Wei, L. Singular value correlation functions
    for products of Wishart random matrices. J. Phys. A. {\bf 46} (2013) 275205.

\bibitem{AkemannIpsenKieburg} Akemann, G.; Ipsen, J.; Kieburg M. Products of rectangular random
    matrices: singular values and progressive scattering.  Phys. Rev. E {\bf 88} (2013) 052118.

\bibitem{AkemannStrahovModels} Akemann, G.; Strahov, E. Product matrix processes for coupled
    multi-matrix models and their hard edge scaling limits. arXiv:1711.01873

\bibitem{BaikDeiftSuidan}Baik, J.; Deift, P.;  Suidan, T. Combinatorics and Random Matrix Theory.
    Graduate Studies in Mathematics, 172.

\bibitem{Baryshnkiov} Baryshnikov, Yu. M. GUEs and queues, Probability Theory and
    Related Fields, 119(2):256--274, 2001.

\bibitem{Beals} Beals, R.; Szmigielski, J. Meijer G-functions: a gentle introduction. Notices Amer.
    Math. Soc. 60 (2013), no. 7, 866--872.


\bibitem{B-det} Borodin, A. Determinantal point processes, in Oxford Handbook of Random
    Matrix Theory, Oxford University Press, 2011. arXiv:0911.1153.

\bibitem{Borodin}
 Borodin, A. Schur dynamics of the Schur processes. Adv. Math. 228 (2011), no. 4, 2268вЂ“-2291

\bibitem{BorodinGorin}
 Borodin, A.; Gorin, V. Lectures on integrable probability. Probability and statistical physics in St. Petersburg, 155--214,
 Proc. Sympos. Pure Math., 91, Amer. Math. Soc., Providence, RI, 2016.

\bibitem{BG_GFF} Borodin, A.; Gorin, V. General beta Jacobi corners process and
    the Gaussian Free Field,
 Communications on Pure and Applied Mathematics, 68, no.\ 10, 1774--1844, (2015). arXiv:1305.3627.


\bibitem{BorodinRains}
 Borodin, A.; Rains, E. M.; Eynard-Mehta theorem, Schur process, and their Pfaffian analogs. J. Stat. Phys. 121 (2005), no. 3--4, 291--317.


\bibitem{BJ} Bougerol, P.; Jeulin, T.; Paths in Weyl chambers and random matrices, Probability
    Theory and Related Fields, 124, no.\ 4 (2002), 517--543.

\bibitem{Bruijn}De Bruijn, N. G.
 On some multiple integrals involving determinants, J. Indian Math.
Soc. (N.S.) 19 (1955), 133–-151.

\bibitem{BG_Schur} Bufetov, A.; Gorin, V.; Fourier transform on high--dimensional unitary groups with
    applications to random tilings, arXiv:1712.09925.


\bibitem{Collins} Collins, B.; Product of random projections, Jacobi ensembles and universality
    problems arising from free probability, Probability Theory and Related Fields, 133, no. 3
    (2005), 315--344, arXiv:math/0406560.

\bibitem{Dimitrov_six} Dimitrov, E.; Six-vertex models and the GUE-corners process, International
    Mathematics Research Notices, rny072, arXiv:1610.06893.

\bibitem{EynardMehta} Eynard, B.; Mehta, M. L. Matrices coupled in a chain. I. Eigenvalue
    correlations. J. Phys. A 31 (1998), no. 19, 4449--4456.



\bibitem{For} Forrester, P.J.; Log-Gases and Random Matrices, Princeton University
    Press, 2010.


\bibitem{ForresterRains} Forrester, P.; Rains, E.; Interpretations of some parameter dependent
    generalizations of classical matrix ensembles, Probability Theory and Related Fields, 131, no.\ 1 (2005),
    1--61. arXiv:math-ph/0211042

\bibitem{Gorin_ASM} Gorin,V.; From Alternating Sign Matrices to the Gaussian Unitary Ensemble,
    Communications in Mathematical Physics, 332, no.\ 1 (2014), 437--447, arXiv:1306.6347.

\bibitem{GP} Gorin, V; Panova.G; Asymptotics of symmetric polynomials with
    applications to statistical mechanics and representation theory, Annals of Probability,
    43, no.\ 6, (2015) 3052--3132. arXiv:1301.0634.

\bibitem{GorinMarcus} Gorin,V; Marcus, A.W.; Crystallization of random matrix orbits, International
    Mathematics Research Notices, rny052, arXiv:1706.07393.

\bibitem{GSun} Gorin, V; Sun, Y; In preparation.

\bibitem{GTW} Gravner, J.; Tracy, C.A.; Widom, H.; Limit Theorems for Height Fluctuations in a
    Class
    of Discrete Space and Time Growth Models, 	J. of Statistical Physics 102 (2001), 1085--1132, arXiv:math/0005133

\bibitem{Ismail}
 Ismail, M. E. H. Classical and quantum orthogonal polynomials in one variable. Encyclopedia of Mathematics and its Applications, 98. Cambridge University Press, Cambridge, 2009.

\bibitem{KieburgKuijlaarsStivigny}
 Kieburg, M.; Kuijlaars, A. B. J.; Stivigny, D. Singular value statistics of matrix products with truncated unitary matrices. Int. Math. Res. Not. IMRN 2016, no. 11, 3392--3424.

\bibitem{KK} Kieburg, M.; Kosters, H; Products of Random Matrices from Polynomial Ensembles,
    arXiv:1601.03724.

\bibitem{KuijlaarsStivigny}
 Kuijlaars, A. B. J.; Stivigny, D. Singular values of products of random matrices and polynomial ensembles. Random Matrices Theory Appl. 3 (2014), no. 3, 1450011.

\bibitem{KuijlaarsZhang} Kuijlaars, A.B.J.; Zhang, L. Singular values of products of Ginibre random
    matrices, multiple orthogonal polynomials and hard edge scaling limits. Commun. Math. Phys.
    {\bf 332} (2014) 759--781.

\bibitem{JohanssonNordenstam} Johansson, K.; Nordenstam, E. Eigenvalues of GUE Minors. Electron. J.
    Probab. 11 (2006), 1342--1371.

\bibitem{Luke}
Luke, Y.L.  The special functions and their approximations. Academic Press, New York 1969.

\bibitem{Macdonald}
 Macdonald, I. G. Symmetric functions and Hall polynomials. Second edition. With contributions by A. Zelevinsky. Oxford Mathematical Monographs. Oxford Science Publications. The Clarendon Press, Oxford University Press, New York, 1995.

\bibitem{Normand}
 Normand, J.-M. Calculation of some determinants using the s-shifted factorial. J. Phys. A 37 (2004), no. 22, 5737--5762.

\bibitem{Jiang} Jiang, T. Approximation of Haar distributed matrices and limiting distributions of
    eigenvalues of Jacobi ensembles, Probab. Theory Related Fields 144 (2009) 221--246.


\bibitem{OCY} O'Connell,N; Yor,M; A Representation for Non-Colliding Random Walks, Electron.
    Commun. Probab. Volume 7 (2002), paper no.\ 1, 1--12.

\bibitem{Ok-wedge} Okounkov, A. Infinite wedge and random partitions. Selecta
    Mathematica 7 (2001), 57--81. arXiv:math/9907127.

\bibitem{OkounkovReshetikhin}
 Okounkov, A.; Reshetikhin, N. Correlation function of Schur process with application to local geometry of a random 3-dimensional Young diagram. J. Amer. Math. Soc. 16 (2003), no. 3, 581--603.


\bibitem{OR-birth} Okounkov, A. Yu.; Reshetikhin, N. Yu. The birth of a
 random matrix, Moscow Mathematical Journal, 6, no.\ 3 (2006), 553--566

\bibitem{OkounkovReshetikhin_skew}
 Okounkov, A.; Reshetikhin, N.
Random Skew Plane Partitions and the Pearcey Process, Communications in Mathematical Physics, 269, no.\ 3 (2007), 571--609.


\bibitem{Olsh_Versh} Olshanski, G.; Vershik, A. Ergodic unitarily invariant
    measures on the space of infinite Hermitian matrices. Contemporary mathematical physics,
    137--175, Amer. Math. Soc. Transl. Ser. 2, 175, Amer. Math. Soc., Providence, RI, 1996,
    arXiv:math/9601215

\bibitem{StrahovD}
 Strahov, E. Dynamical correlation functions for products of random matrices. Random Matrices Theory Appl. 4 (2015), no. 4, 1550020.


\bibitem{TW}  Tracy, C.; Widom, H. Level-spacing distributions and the Airy
    kernel, Communications in Mathematical Physics, 159, no.\ 1 (1994), 151--174.

\bibitem{ZhangChen} Zhang, R.; Chen, Li-Chen. Matrix inversion using orthogonal polynomials. Arab
    J. Math. Sci. 17 (2011), no. 1, 11--30.
\end{thebibliography}
\end{document}